\documentclass[reqno,11pt]{amsart}
\usepackage[utf8]{inputenc}
\usepackage{graphicx}
\usepackage{slashed}
\usepackage{amscd}
\usepackage{amssymb}
\usepackage{esint}
\usepackage{enumitem}
\usepackage{comment}
\usepackage{amsmath}
\usepackage{dsfont}
\usepackage[mathscr]{eucal}
\usepackage{pst-all}

\numberwithin{equation}{section}
\allowdisplaybreaks[1]

\title[Elliptic Methods for Causal Variational Principles]{Elliptic Methods for Solving the Linearized Field Equations of Causal Variational Principles}

\author[F.\ Finster]{Felix Finster}
\author[M.\ Lottner]{Magdalena Lottner \\ \\ November 2021}
\address{Fakult\"at f\"ur Mathematik \\ Universit\"at Regensburg \\ D-93040 Regensburg \\ Germany}
\email{finster@ur.de, magdalena.lottner@ur.de}
\newtheorem{Def}{Definition}[section]
\newtheorem{Thm}[Def]{Theorem}
\newtheorem{Prp}[Def]{Proposition}
\newtheorem{Lemma}[Def]{Lemma}
\newtheorem{Remark}[Def]{Remark}

\newcommand{\Thanks}{\vspace*{.5em} \noindent \thanks}
\newcommand{\beq}{\begin{equation}}
\newcommand{\eeq}{\end{equation}}
\newcommand{\Proof}{\begin{proof}}
	\newcommand{\QED}{\end{proof} \noindent}
\newcommand{\QEDrem}{\ \hfill $\Diamond$}

\newcommand{\la}{\langle}
\newcommand{\ra}{\rangle}
\newcommand{\lla}{\langle\!\langle}
\newcommand{\rra}{\rangle\!\rangle}

\newcommand{\C}{\mathbb{C}}
\newcommand{\R}{\mathbb{R}}

\newcommand{\Z}{\mathbb{Z}}
\newcommand{\N}{\mathbb{N}}

\DeclareMathOperator{\tr}{tr}

\renewcommand{\L}{{\mathcal{L}}}
\newcommand{\Sact}{{\mathcal{S}}}

\renewcommand{\H}{\mathscr{H}}

\newcommand{\Lin}{\text{\rm{L}}}
\newcommand{\F}{{\mathscr{F}}}

\DeclareMathOperator{\norm}{| \hspace*{-0.1em}| \hspace*{-0.1em}|}

\newcommand{\D}{\mathscr{D}}

\DeclareMathOperator{\supp}{supp}

\newcommand{\scrN}{\mycal N}
\newcommand{\itemD}{\item[{\raisebox{0.125em}{\tiny $\blacktriangleright$}}]}
\newcommand{\scrU}{{\mathscr{U}}}

\newcommand{\J}{\mathfrak{J}}
\newcommand{\s}{\mathfrak{s}}
\newcommand{\Jdiff}{\mathfrak{J}^\text{\rm{\tiny{diff}}}}
\newcommand{\Jtest}{\mathfrak{J}^\text{\rm{\tiny{test}}}}
\newcommand{\Jvary}{\mathfrak{J}^\text{\rm{\tiny{vary}}}}
\newcommand{\Jin}{\mathfrak{J}^\text{\rm{\tiny{in}}}}
\newcommand{\Jlin}{\mathfrak{J}^\text{\rm{\tiny{lin}}}}
\newcommand{\Gdiff}{\Gamma^\text{\rm{\tiny{diff}}}}
\newcommand{\Gtest}{\Gamma^\text{\rm{\tiny{test}}}}
\newcommand{\Ctest}{C^\text{\rm{\tiny{test}}}}
\renewcommand{\u}{\mathfrak{u}}
\renewcommand{\v}{\mathfrak{v}}
\newcommand{\w}{\mathfrak{w}}
\newcommand{\bitem}{\begin{itemize}[leftmargin=2.5em]}
\newcommand{\eitem}{\end{itemize}}
\renewcommand{\div}{{\rm{div}}\,}

\newcommand{\G}{{\mathscr{G}}}
\newcommand{\x}{\mathbf{x}}
\newcommand{\y}{\mathbf{y}}

\newcommand{\h}{\mathfrak{h}}

\DeclareFontFamily{OT1}{rsfso}{}
\DeclareFontShape{OT1}{rsfso}{m}{n}{ <-7> rsfso5 <7-10> rsfso7 <10-> rsfso10}{}
\DeclareMathAlphabet{\mycal}{OT1}{rsfso}{m}{n}

\begin{document}
\maketitle
\begin{abstract}
The existence theory is developed for solutions of the inhomogeneous linearized field equations
for causal variational principles. These equations are formulated weakly with an integral operator
which is shown to be bounded and symmetric on a Hilbert space endowed with a suitably adapted
weighted $L^2$-scalar product.
Guided by the procedure in the theory of linear elliptic partial differential equations, we
use the spectral calculus to define Sobolev-type Hilbert spaces
and invert the linearized field operator as an operator between such function spaces.
The uniqueness of the resulting weak solutions is analyzed.
Our constructions are illustrated in simple explicit examples.
The connection to the causal action principle for static causal fermion systems is explained.
\end{abstract}
\tableofcontents

\section{Introduction}
Causal variational principles were introduced in~\cite{continuum}
as a mathematical generalization of the causal action principle,
being the analytical core of the physical theory of causal fermion systems
(for the general context see the reviews~\cite{ nrstg, review, dice2014}, the textbooks~\cite{cfs, intro}
or the website~\cite{cfsweblink}).
In general terms, given a manifold~$\G$ together with a non-negative function~$\L : \G \times \G \rightarrow \R^+_0$,
in a {\em{causal variational principle}} one minimizes the action~$\Sact$ given by
\[ 
\Sact (\mu) = \int_\G d\mu(\x) \int_\G d\mu(\y)\: \L(\x,\y) \]
under variations of the measure~$\mu$ on~$\G$, keeping the total volume~$\mu(\G)$ fixed
(for the precise mathematical setup see Section~\ref{seccvp} below).
The support of the measure~$\mu$ denoted by
\[ N := \supp \mu \subset \G \]
has the interpretation as the underlying {\em{space}} or {\em{spacetime}}.
A minimizing measure satisfies corresponding {\em{Euler-Lagrange (EL) equations}}
(for details see the preliminaries in Section~\ref{seccvp}
or~\cite[Chapter~7]{intro}).
For the detailed analysis of minimizing measures, it is very useful to consider first variations of the measure~$\mu$
which preserve the EL equations. Such a variation is described by 
a so-called {\em{jet}}~$\v=(b,v)$, which consists of a scalar function~$b$ and a vector field~$v$
(for details see~\eqref{Jdef} in Section~\ref{secwEL} below) satisfying
the {\em{homogeneous linearized field equations}}
\[ \Delta \v = 0 \:, \]
where the operator~$\Delta$ is defined by
\[ \Delta \v(\x) = \nabla \bigg( \int_N \big( \nabla_{1, \v} + \nabla_{2, \v} \big) \L(\x,\y)\: d\mu(\y) - \nabla_\v \,\s \bigg) \:. \]
Here~$\s$ is a positive parameter, and the jet derivative~$\nabla$ is a combination of
multiplication and directional derivative (for details see again Section~\ref{secwEL}).
For the mathematical analysis of the linearized field equations, it is preferable to
include an inhomogeneity~$\w$,
\beq \label{Delvw}
\Delta \v = \w \:.
\eeq
In view of the formal similarity to the Poisson equation, we sometimes refer to the operator~$\Delta$
as the {\em{Laplacian}}. 
Despite this analogy, one
should keep in mind that the {\em{inhomogeneous linearized field equations
are not differential equations}},
but instead they are nonlocal equations involving integrals of specific integral kernels.
Nevertheless, it turns out that methods of partial differential equations (PDEs) can be used for their analysis.
This has first been explored in~\cite{linhyp}, where methods of hyperbolic PDEs
were used to prove under general assumptions that the Cauchy problem is well-posed.
In the present paper, we explore how methods of the theory of {\em{elliptic}} PDEs are applicable to the analysis
of the linearized field equations.
The main application in mind are static linearized fields in static causal fermion systems.
In analogy to the scalar wave equation, which for time-independent fields goes over to the
Poisson equation, one might expect that the linearized field equations of the underlying
causal variational principle might change from hyperbolic to elliptic type.
It is the goal of this paper to show that this naive expectation is indeed correct.
Moreover, we work out abstractly how elliptic estimates can be used for the
analysis of the linearized field equations in the static setting.
To this end, we realize the Laplacian as a bounded symmetric operator on
a weighted~$L^2$-space. Using the spectral calculus, we define
Sobolev-type Hilbert spaces and prove that the Laplacian, considered as an operator
between such function spaces, has a well-defined inverse.
Our abstract constructions are illustrated in various examples.

The paper is organized as follows. Section~\ref{secprelim} provides the necessary preliminaries
on causal variational principles and the linearized field equations.
In Section~\ref{secelliptic}, the bilinear form obtained by weak evaluation of the linearized field equations
is written as~$\lla \u, \Delta_N \v \rra$, where~$\lla .,. \rra$ is a weighted $L^2$-scalar product
(with a suitable ``adapted'' weight function~$h$ constructed out of the Lagrangian), and~$\Delta_N$
is a bounded symmetric operator on the resulting Hilbert space.
In Section~\ref{secpoisson}, we proceed by solving the inhomogeneous linearized field equations.
In Section~\ref{secexstatic} we illustrate our constructions in simple explicit examples.
In Section~\ref{seccfsstatic} we explain the connection to static linearized fields in static causal
fermion systems.
Finally, in Appendix~\ref{appnoscalar} the role of the scalar components of the jets is clarified.

\section{Preliminaries} \label{secprelim}
We now recall the basics on causal variational principles in the setting
needed here. More details can be found in~\cite{jet, intro}.
We use a slightly different notation in order to get consistency with the causal variational
principle in the static case to be introduced in Section~\ref{seccfsstatic}.

\subsection{Causal Variational Principles in the Non-Compact Setting} \label{seccvp}
We consider causal variational principles in the non-compact setting as first
introduced in~\cite[Section~2]{jet}. Thus we let~$\G$ be a (possibly non-compact)
smooth manifold of dimension~$m \geq 1$
and~$\mu$ a (positive) Borel measure on~$\G$.
Moreover, we are given a non-negative function~$\L : \G \times \G \rightarrow \R^+_0$
(the {\em{Lagrangian}}) with the following properties:
\begin{itemize}[leftmargin=2em]
\item[(i)] $\L$ is symmetric: $\L(\x,\y) = \L(\y,\x)$ for all~$\x,\y \in \G$.\label{Cond1}
\item[(ii)] $\L$ is lower semi-continuous, i.e.\ for all sequences~$\x_n \rightarrow \x$ and~$\y_{n'} \rightarrow \y$,
\label{Cond2}%
\[ \L(\x,\y) \leq \liminf_{n,n' \rightarrow \infty} \L(\x_n, \y_{n'})\:. \]
\end{itemize}
The {\em{causal variational principle}} is to minimize the action
\beq \label{Sact} 
\Sact (\mu) = \int_\G d\mu(\x) \int_\G d\mu(\y)\: \L(\x,\y) 
\eeq
under variations of the measure~$\mu$, keeping the total volume~$\mu(\G)$ fixed
({\em{volume constraint}}).

If the total volume~$\mu(\G)$ is finite, one minimizes the causal action within the class of
all regular Borel measures with fixed total volume.
If the total volume~$\mu(\G)$ is infinite, however, it is not obvious how to implement the volume constraint,
making it necessary to proceed as follows.
We make the following additional assumptions:
\begin{itemize}[leftmargin=2em]
\item[(iii)] The measure~$\mu$ is {\em{locally finite}}
(meaning that any~$\x \in \G$ has an open neighborhood~$U$ with~$\mu(U)< \infty$)
and {\em{regular}} (meaning that the measure of a set can be recovered by approximation from inside
with compact and from outside with open sets). \label{Cond3}
\item[(iv)] The function~$\L(\x,.)$ is $\mu$-integrable for all~$\x \in \G$, giving
a lower semi-continuous and bounded function on~$\G$. \label{Cond4}
\end{itemize}
Given a regular Borel measure~$\mu$ on~$\G$, we vary over all
regular Borel measures~$\tilde{\mu}$ with
\beq \label{totvol}
\big| \tilde{\mu} - \mu \big|(\G) < \infty \qquad \text{and} \qquad
\big( \tilde{\mu} - \mu \big) (\G) = 0
\eeq
(where~$|.|$ denotes the total variation of a measure).
For such variations, the difference of the actions is well-defined by
\begin{align}
\Sact(\tilde{\mu}) - \Sact(\mu) =& \int_\G d(\tilde{\mu}-\mu)(\x) \int_\G d\mu(\y) \:\L(\x,\y) + \int_\G d\mu(\x) \int_\G d(\tilde{\mu}-\mu)(\y)\: \L(\x,\y) \notag \\
&+  \int_\G d(\tilde{\mu}-\mu)(\x) \int_\G d(\tilde{\mu}-\mu)(\y) \:\L(\x,\y) \:. \label{Sdiff}
\end{align}
The measure~$\mu$ is said to be a {\em{minimizer under variations of finite volume}}
if this difference is non-negative for all regular Borel measures satisfying~\eqref{totvol}.
The existence theory for such minimizers is developed in~\cite{noncompact}.
Moreover, it is shown in~\cite[Lemma~2.3]{jet} that a minimizer
satisfies the {\em{Euler-Lagrange (EL) equations}}
which state that for a suitable value of the parameter~$\s>0$,
the lower semi-continuous function~$\ell : \G \rightarrow \R_0^+$ defined by
\beq \label{elldef}
\ell(\x) := \int_\G \L(\x,\y)\: d\mu(\y) - \s
\eeq
is minimal and vanishes on the support of~$\mu$,
\beq \label{EL}
\ell|_{\supp \mu} \equiv \inf_\G \ell = 0 \:.
\eeq
For further details we refer to~\cite[Section~2]{jet}.

\subsection{The Restricted Euler-Lagrange Equations and Jet Spaces} \label{secwEL}
\hspace*{0.05cm}
We denote the support of~$\mu$ by~$N$,
\[ 
N:= \supp \mu \;\subset\; \G \:. \]
The EL equations~\eqref{EL} are nonlocal in the sense that
they make a statement on~$\ell$ even for points~$\x \in \G$ which
are far away from~$N$.
It turns out that for the applications in this paper, it is preferable to
evaluate the EL equations only locally in a neighborhood of~$N$.
This leads to the {\em{restricted EL equations}} introduced in~\cite[Section~4]{jet}.
We here give a slightly less general version of these equations which
is sufficient for our purposes. In order to explain how the restricted EL equations come about,
we begin with the simplified situation that the function~$\ell$ is smooth.
In this case, the minimality of~$\ell$ implies that the derivative of~$\ell$
vanishes on~$N$, i.e.\
\beq \label{ELrestricted}
\ell|_N \equiv 0 \qquad \text{and} \qquad D \ell|_N \equiv 0
\eeq
(where~$D \ell(p) : T_p \G \rightarrow \R$ is the derivative).
In order to combine these two equations in a compact form,
it is convenient to consider a pair~$\u := (a, u)$
consisting of a real-valued function~$a$ on~$N$ and a vector field~$u$
on~$T\G$ along~$N$, and to denote the combination of 
multiplication and directional derivative by
\beq \label{Djet}
\nabla_{\u} \ell(\x) := a(\x)\, \ell(\x) + \big(D_u \ell \big)(\x) \:.
\eeq
Then the equations~\eqref{ELrestricted} imply that~$\nabla_{\u} \ell(\x)$
vanishes for all~$\x \in N$.
The pair~$\u=(a,u)$ is referred to as a {\em{jet}}.

In the general lower-continuous setting, one must be careful because
the directional derivative~$D_u \ell$ in~\eqref{Djet} need not exist.
Our method for dealing with this issue is to restrict attention to vector fields
for which the directional derivative is well-defined.
Moreover, we must specify the regularity assumptions on~$a$ and~$u$.
To begin with, we always assume that~$a$ and~$u$ are {\em{smooth}} in the sense that they
have a smooth extension to the manifold~$\G$. Thus the jet~$\u$ should be
an element of the jet space
\beq \label{Jdef}
\J := \big\{ \u = (a,u) \text{ with } a \in C^\infty(N, \R) \text{ and } u \in \Gamma(N, T\G) \big\} \:,
\eeq
where~$C^\infty(N, \R)$ and~$\Gamma(N,T\G)$ denote the space of real-valued functions and vector fields
on~$N$, respectively, which admit a smooth extension to~$\G$.
We always denote spaces of compactly supported jets with a subscript~$0$; for example,
\beq \label{J0def}
\J_0 := \big\{ \u \in \J \:\big|\: \text{$\supp \u$ compact}  \big\} \:.
\eeq

Clearly, the fact that a jet~$\u$ is smooth does not imply that the functions~$\ell$
or~$\L$ are differentiable in the direction of~$\u$. This must be ensured by additional
conditions which are satisfied by suitable subspaces of~$\J$
which we now introduce.
First, we let~$\Gdiff$ be those vector fields for which the
directional derivative of the function~$\ell$ exists,
\[ 
\Gdiff = \big\{ u \in C^\infty(N, T\G) \;\big|\; \text{$D_{u} \ell(\x)$ exists for all~$\x \in N$} \big\} \:. \]
This gives rise to the jet space
\beq \label{Jdiffdef}
\Jdiff := C^\infty(N, \R) \oplus \Gdiff \;\subset\; \J \:.
\eeq
For the jets in~$\Jdiff$, the combination of multiplication and directional derivative
in~\eqref{Djet} is well-defined. 
We choose a linear subspace~$\Jtest \subset \Jdiff$ with the properties
that its scalar and vector components are both vector spaces,
\[ 
\Jtest = \Ctest(N, \R) \oplus \Gtest \;\subseteq\; \Jdiff \:, \]
and that the scalar component is nowhere trivial in the sense that\footnote{This assumption
is convenient, because then the restricted EL equations~\eqref{ELtest}
imply that~$\ell$ vanishes identically on~$M$.}
\beq \label{Cnontriv}
\text{for all~$\x \in N$ there is~$a \in \Ctest(N, \R)$ with~$a(\x) \neq 0$}\:.
\eeq
Then the {\em{restricted EL equations}} read (for details cf.~\cite[(eq.~(4.10)]{jet})
\beq \label{ELtest}
\nabla_{\u} \ell|_N = 0 \qquad \text{for all~$\u \in \Jtest$}\:.
\eeq
We remark that, in the literature, the restricted EL equations are sometimes also referred to as
the {\em{weak}} EL equations. Here we prefer the notion ``restricted'' in order to avoid confusion
with weak solutions of these equations as constructed in~\cite{linhyp}
(in the static setting, such weak solutions will be constructed in Section~\ref{secpoisson}).
The purpose of introducing~$\Jtest$ is that it gives the freedom to restrict attention to the part of
information in the EL equations which is relevant for the application in mind.

We conclude this section by introducing spaces of jets with suitable differentiability properties.
Before beginning, we point out that, here and throughout this paper, we use the following conventions for partial derivatives and jet derivatives:
\begin{itemize}[leftmargin=2em]
\itemD Partial and jet derivatives with an index $i \in \{ 1,2 \}$, as for example in~\eqref{derex}, only act on the respective variable of the function $\L$.
This implies, for example, that the derivatives commute,
\[ 
\nabla_{1,\v} \nabla_{1,\u} \L(\x,\y) = \nabla_{1,\u} \nabla_{1,\v} \L(\x,\y) \:. \]
\itemD The partial or jet derivatives which do not carry an index act as partial derivatives
on the corresponding argument of the Lagrangian. This implies, for example, that
\beq \label{conpartial}
\nabla_\u \int_\G \nabla_{1,\v} \, \L(\x,\y) \: d\mu(\y) =  \int_\G \nabla_{1,\u} \nabla_{1,\v}\, \L(\x,\y) \: d\mu(\y) \:.
\eeq
\end{itemize}
Thus {\em{jets are never differentiated}}. We now introduce the spaces~$\J^\ell$, where~$\ell \in \N_0 \cup \{\infty\}$ can be
thought of as the order of differentiability if the derivatives act  simultaneously on
both arguments of the Lagrangian:
\begin{Def} \label{defJvary}
For any~$\ell \in \N_0 \cup \{\infty\}$, the jet space~$\J^\ell \subset \J$
is defined as the vector space of test jets with the following properties:
\begin{itemize}[leftmargin=2em]
\item[\rm{(i)}] For all~$\y \in N$ and all~$\x$ in an open neighborhood of~$N$,
directional derivatives
\beq \label{derex}
\big( \nabla_{1, \v_1} + \nabla_{2, \v_1} \big) \cdots \big( \nabla_{1, \v_p} + \nabla_{2, \v_p} \big) \L(\x,\y)
\eeq
(computed componentwise in charts around~$\x$ and~$\y$)
exist for all~$p \in \{1, \ldots, \ell\}$ and all~$\v_1, \ldots, \v_p \in \J^\ell$.
\item[\rm{(ii)}] The functions in~\eqref{derex} are $\mu$-integrable
in the variable~$\y$, giving rise to locally bounded functions in~$\x$. More precisely,
these functions are in the space
\[ L^\infty_\text{\rm{loc}}\Big( N, L^1\big(N, d\mu(\y) \big); d\mu(\x) \Big) \:. \]
\item[\rm{(iii)}] Integrating the expression~\eqref{derex} in~$\y$ over~$N$
with respect to the measure~$\mu$,
the resulting function (defined for all~$\x$ in an open neighborhood of~$N$)
is continuously differentiable in the direction of every jet~$\u \in \Jtest$.
\end{itemize}
\end{Def} \noindent

\subsection{The Linearized Field Equations} \label{seclinear}

The EL equations~\eqref{EL} (and similarly the restricted EL equations~\eqref{ELrestricted})
are nonlinear equations because they involve the measure~$\rho$
in a twofold way: first, the measure comes up as the integration measure in~\eqref{elldef},
and second the function~$\ell$ is evaluated on the support of this measure.
Following the common procedure in mathematics and physics, one can simplify the
problem by considering linear perturbations about a given solution.
As an example in classical field theory, considering a family~$(g_\tau)_{\tau \in [0,1)}$ of
Lorentzian metrics which all satisfy the vacuum Einstein equations, the linearization~$\partial_\tau g_\tau$
describes gravitational waves propagating in the spacetime with metric~$g_\tau$.
The analogous notion in the setting of causal fermion systems
is a linearization of a family of measures~$(\tilde{\mu}_\tau)_{\tau \in [0,1)}$ which all satisfy the restricted EL equations~\eqref{ELtest}
(for fixed values of the parameters~$\kappa$ and~$\s$).
It turns out to be fruitful to construct this family of measures by multiplying
a given minimizing measure~$\mu$ by a weight function~$f_\tau$ and then
``transporting'' the resulting measure with a mapping~$F_\tau$. More precisely, we consider the ansatz
\begin{align} \label{rhoFf}
\tilde{\mu}_\tau = (F_\tau)_* \big( f_\tau \, \mu \big) \:,
\end{align}
where~$f_\tau \in C^\infty(N, \R^+)$ and~$F_\tau \in C^\infty(N, \G)$ are smooth mappings,
and~$(F_\tau)_*\mu$ denotes the push-forward (defined 
for a subset~$\Omega \subset \G$ by~$((F_\tau)_*\mu)(\Omega)
= \mu ( F_\tau^{-1} (\Omega))$; see for example~\cite[Section~3.6]{bogachev}).

The property of the family of measures~$(\tilde{\mu}_\tau)_{\tau \in [0,1)}$ of the form~\eqref{rhoFf}
to satisfy the restricted EL equation for all~$\tau$
means infinitesimally in~$\tau$ that the jet~$\v$ defined by
\beq \label{vinfdef}
\v = (b,v) := \frac{d}{d\tau} (f_\tau, F_\tau) \big|_{\tau=0}
\eeq
satisfies the {\em{linearized field equations}}. We now
recall the main step of the construction.
Using the definition of the push-forward measure, we can write the restricted EL equations~\eqref{ELtest}
for the measure~$\tilde{\mu}_\tau$ as
\[ \nabla_\u \bigg( \int_N \L\big( F_\tau(\x), F_\tau(\y) \big) \: f_\tau(\y)\: d\mu - \s \bigg) = 0 \:. \]
Since the function~$\ell$ vanishes on the support, we may multiply by~$f_\tau(\x)$ to obtain
\beq \label{lin1}
\nabla_\u \bigg( \int_N f_\tau(\x)\: \L\big( F_\tau(\x), F_\tau(\y) \big) \: f_\tau(\y)\: d\mu - f_\tau(\x)\,\s \bigg) = 0 \:.
\eeq
At this point, the technical complication arise that one must specify the $\tau$-dependence of the jet spaces,
and moreover the last transformation makes it necessary to transform the jet spaces.
Here we do not enter the details but refer instead to the rigorous derivation in~\cite[Section~3.3]{perturb}
or to the simplified presentation in the smooth setting in the textbook~\cite[Chapter~6]{intro}.
Differentiating~\eqref{lin1} with respect to~$\tau$ gives the homogeneous linearized field equations 
\beq \label{eqlinlip}
\la \u, \Delta \v \ra|_N = 0 \qquad \text{for all~$\u \in \Jtest$} \:,
\eeq
where
\beq \label{Lapdef}
\la \u, \Delta \v \ra(\x) := \nabla_{\u} \bigg( \int_N \big( \nabla_{1, \v} + \nabla_{2, \v} \big) \L(\x,\y)\: d\mu(\y) - \nabla_\v \,\s \bigg) \:.
\eeq
We denote the vector space of all solutions of the linearized field equations by~$\Jlin \subset \J^1$.

\subsection{Positive Functionals Arising from Second Variations} \label{secsecond}
Another ingredient to our constructions are positive functionals which arise in the
analysis of second variations~\cite{positive}. We now recall a few concepts and results.

Clearly, if~$\mu$ is a {\em{minimizing}} measure, then second variations are non-negative.
For our purposes, it again suffices to consider variations of the form~\eqref{rhoFf},
where for simplicity we assume that~$f_\tau$ and~$F_\tau$ are trivial outside a compact set.
Under these assumptions, it is proven in~\cite[Theorem~1.1]{positive} that the positivity
of second variations gives rise to the inequality
\[ \int_N d\mu(\x) \int_N d\mu(\y) \:\nabla_{1,\v} \nabla_{2,\v} \L(\x,\y) 
+ \int_N \nabla^2 \ell|_\x(\v,\v)\: d\mu(\x) \geq 0 \:, \]
where jet~$\v \in \J_0$ is again an infinitesimal generator of a variation of the form~\eqref{vinfdef},
which, however, does not need to respect the EL equations
(and~$\J_0$ are the compactly supported jets~\eqref{J0def}).
For our purposes, it is preferable to write this inequality as
\beq \label{posint}
\frac{1}{2} \int_N d\mu(\x) \int_N d\mu(\y) \:
\big( \nabla_{1,\v} + \nabla_{2,\v} \big)^2 \L(\x,\y) - \int_N (\nabla^2 \s)(\v,\v)\: d\mu \geq 0 \:.
\eeq
Then it is obvious that the integrals are well-defined if we assume that~$\u,\v \in \J^2$
(see Definition~\ref{defJvary}).
Moreover, using~\eqref{elldef} and~\eqref{Lapdef}, the inequality can be written in the compact form
\beq \label{vvpositive}
\la \v, \Delta \v \ra_N \geq 0 \qquad \text{for all~$\v \in \J^2_0$} \:,
\eeq
where we used the notation
\beq \label{uLapv}
\la \u, \Delta \v \ra_N := \int_N \la \u, \Delta \v \ra(\x)\: d\mu(\x) \:.
\eeq
In other words, the operator representing~$\la ., \Delta . \ra_N$ is {\em{positive semi-definite}}.
This positivity property is not an assumption, but it follows already from the
structure of causal variational principles.
It might come as a surprise, because the analogous inequality for the wave operator
in Minkowski space is violated. Instead, this inequality holds (up to an irrelevant sign)
for the Laplacian in the {\em{Riemannian}} setting. These facts are not a contradiction if
one keeps in mind that the operator~$\Delta$ has a structure which is
very different from a differential operator. The basic reason why~\eqref{vvpositive} holds is that, in the
setting of causal variational principles, we consider minimizers. In contrast, the Dirichlet
energy in the hyperbolic setting is unbounded from below, making it necessary to work instead with
critical points. The fact that, in the theory of causal fermion systems, 
the dynamics in spacetime is described by minimizers (and not merely critical points) is a specific
feature of the causal action principle. We refer the interested reader to the text books~\cite{cfs, intro}.

\section{The Laplacian as a Bounded Symmetric Operator} \label{secelliptic}
\subsection{The Adapted Weighted $L^2$-Scalar Product}
In preparation, we let~$\Gamma_\x$ be the subspace of the tangent space spanned by the test jets,
\[ 
\Gamma_\x := \big\{ u(\x) \:|\: u \in \Gtest \big\} \;\subset\; T_\x\G\:. \]
We introduce a Riemannian metric~$g_\x$ on~$\Gamma_\x$.
This Riemannian metric also induces a pointwise scalar product on the jets. Namely, setting
\[ 
\J_\x := \R \oplus \Gamma_\x \:, \]
we obtain the scalar product on~$\J_\x$
\beq
\la .,. \ra_\x \,:\, \J_\x \times \J_\x \rightarrow \R \:,\qquad
\la \v, \tilde{\v} \ra_\x := b(\x)\, \tilde{b}(\x) + g_\x \big(v(\x),\tilde{v}(\x) \big) \:, \label{vsprod}
\eeq
were the scalar and vector components of~$\v$ and~$\tilde{\v}$
are denoted by~$\v=(v,b), \tilde{\v}=(\tilde{v},\tilde{b})$.
We denote the corresponding norm by~$\|.\|_\x$.
We note for clarity that the choice of the Riemannian metric~$g_x$ involves a certain degree of freedom.
This freedom can be used to our advantage in order to help to arrange the technical assumptions needed below
(in particular that the norm~$\|\overline{\nabla}^2 \L(x,y)\|$ in~\eqref{D2L} be finite almost everywhere).
We point out that the subsequent constructions may depend on the choice of~$g_x$.

In order to have the largest possible flexibility, we shall work with a subspace~$\Jvary$
of the compactly supported test jets,
\beq \label{Jvarydef}
\Jvary \subset \Jtest \cap \J^2_0 \:,
\eeq
which we can choose arbitrarily depending on the application in mind.
In particular, the scalar component of~$\Jvary$
does not need to be nontrivial in the sense~\eqref{Cnontriv}.
Indeed, a case of particular interest is when the~$\Jvary$ has no scalar component,
as will be explained in detail in Appendix~\ref{appnoscalar}.
We now consider the bilinear form (see~\eqref{uLapv})
\beq \label{Delbilin}
\la \,.\, , \Delta \,.\, \ra_N \::\: \Jvary \times \Jvary \rightarrow \R \:.
\eeq
This bilinear form is {\em{positive semi-definite}} according to~\eqref{vvpositive}.
Using the abbreviation
\[ \overline{\nabla}_{\w}  := \big( \nabla_{1,\w} + \nabla_{2,\w} \big) \:, \]
we write~\eqref{posint} as
\[ \la \v, \Delta \v \ra_N = \frac{1}{2} \int_N d\mu(\x) \int_N d\mu(\y) \:
\overline{\nabla}_\v \overline{\nabla}_\v \L(\x,\y) - \int_N (\nabla^2 \s)(\v,\v)\: d\mu \geq 0 \]
(and~$(\nabla^2 \s)(\v,\v) = \nabla_\v \nabla_\v \s = b(x)^2\, \s$).
This shows in particular that the bilinear form~$\la \,.\, , \Delta \,.\, \ra_N$ is {\em{symmetric}}.
For all~$\x, \y \in N$ we define
\beq \label{D2L}
\big\|\overline{\nabla}^2 \L(\x,\y) \big\| := \sup_{\u, \v \in \Jvary}
\frac{\big|\overline{\nabla}_\u \overline{\nabla}_\v \L(\x,\y) \big|}{\big(\|\u(\x)\|_\x + \|\u(\y)\|_\y \big)
\big( \|\v(\x)\|_\x + \|\v(\y)\|_\y \big)} \:,
\eeq
where the supremum is taken only over those jets for which the denominator is non-zero. Furthermore,
we assume that the Lagrangian is sufficiently regular such that this quantity is finite
for almost all~$\x,\y \in N$.
Then for any $\v=(v,b) \in \Jvary$:
\begin{align*}
\big\| \overline{\nabla}_\v \overline{\nabla}_\v \L(\x,\y) \big\|
&\leq \big\| \overline{\nabla}^2 \L(\x,\y) \big\|\: \big(\|\v(\x)\|_\x + \|\v(\y)\|_\y \big)^2 \\
&\leq 2\:\big\| \overline{\nabla}^2 \L(\x,\y) \big\|\: \big(\|\v(\x)\|_\x^2 + \|\v(\y)\|_\y^2 \big) \:.
\end{align*}
We thus obtain the estimate
\begin{align}
\la \v, \Delta \v \ra_N &\leq 
\int_N d\mu(\x) \int_N d\mu(\y) \:
\big\| \overline{\nabla}^2 \L(\x,\y) \big\|\: \big(\|\v(\x)\|_\x^2 + \|\v(\y)\|_\y^2 \big)- \int_N b(\x)^2\: \s\: d\mu(\x) \notag \\
&= 2 \int_N d\mu(\x) \|\v(\x)\|_\x^2 \int_N d\mu(\y) \:
\big\| \overline{\nabla}^2 \L(\x,\y) \big\| 
- \int_N b(\x)^2\: \s\: d\mu(\x) \:. \label{Deles}
\end{align}

\begin{Def} \label{defadapted} The {\bf{adapted weighted $L^2$-scalar product}} is defined by
\beq \label{normadapt}
\lla \u, \v \rra := \int_N \la \u(\x),\v(\x) \ra_\x\: h(\x)\: d\mu(\x) \:,
\eeq
where~$h$ is the weight function
\beq \label{nudef}
h(\x) := 1 + \int_N \big\| \overline{\nabla}^2 \L(\x,\y) \big\| \: d\mu(\y) \:.
\eeq
The corresponding norm is denoted by~$\norm . \norm$.
\end{Def}

We assume that this norm is finite for all jets in~$\Jvary$, thus defining a scalar product on~$\Jvary$,
\beq \label{ssprod}
\lla .,. \rra \::\: \Jvary \times \Jvary \rightarrow \R\:.
\eeq
This assumption is satisfied for example if we assume that the
integral in~\eqref{nudef} is a locally bounded function in~$\x$
(note that the jets in~$\Jvary$ are all compactly supported).
If the function~$h$ is unbounded, one can still
arrange in many situations that the scalar product~\eqref{ssprod} is finite by adapting~$\Jvary$ appropriately
(for instance by working with compactly supported jets; see the example in Section~\ref{secweightnontriv}).
Taking the completion, we obtain the Hilbert space~$(\h, \lla .|. \rra)$.
By definition, the weighted norm is larger than the $L^2$-norm, i.e.\
\[ \|\u\|_{L^2(N, d\mu)} \leq \norm \u \norm \qquad \text{for all~$\u \in \Jvary$} \:, \]
giving rise to a natural embedding~$\h \hookrightarrow L^2(N, d\mu)$.

\subsection{Spectral Decomposition of the Laplacian}
The estimate~\eqref{Deles} shows that the bilinear form~$\la .,\Delta . \ra_N$ is bounded with
respect to the adapted weighted $L^2$-scalar product introduced in Definition~\ref{defadapted},
i.e.\ there is a constant~$C>0$ such that
\[ \big| \la \u, \Delta \v \ra_N \big| \overset{(*)}{\leq} \sqrt{\la \u, \Delta \u \ra_N}\:
\sqrt{\la \v, \Delta \v \ra_N}
\leq C \norm \u \norm \norm \v \norm 
\qquad \text{for all~$\u, \v \in \Jvary$}\:, \]
where in~$(*)$ we applied the Cauchy-Schwarz inequality to the positive semi-definite
bilinear form~$\la ., \Delta . \ra$ (see~\eqref{vvpositive}).
The Fr{\'e}chet-Riesz theorem yields a uniquely determined bounded and symmetric operator~$\Delta_N$
with the property that
\beq \label{DelNdef}
\la \u, \Delta \v \ra_N = \lla \u, \Delta_N\, \v \rra \qquad \text{for all~$\u, \v \in \h$}\:.
\eeq

We let~$\h^\C$ be the complexification of the vector space~$\h$, i.e.\
\[ \h^\C = \h \oplus i \h \]
with the obvious multiplication by complex numbers. 
Extending the bilinear form~$\lla .|. \rra$ to a sesquilinear form on~$\h^\C$ by
\[ \lla (\u+i\v) \,|\, (\mathfrak{z}+i\mathfrak{w}) \rra := \lla \u \,|\, \mathfrak{z} \rra
-i\, \lla \v \,|\, \mathfrak{z} \rra
+ i \lla \u\,|\, \mathfrak{w} \rra+\lla \v \,|\, \mathfrak{w} \rra \:, \]
we obtain the complex Hilbert space~$(\h^\C, \lla .|. \rra)$.
The Laplacian~$\Delta_N$ extends to a complex-linear bounded operator on~$\h^\C$ by
\[ \Delta_N (\u+i\v) := \Delta_N \u + i \Delta_N \v \:. \]
The resulting operator is again symmetric and positive semi-definite, i.e.\ 
\[ \lla \u \,|\, \Delta_N\, \v \rra = \lla \Delta_N\,  \u \,|\, \v \rra \qquad \text{and} \qquad
\lla \u \,|\, \Delta_N\, \u \rra \geq 0 \qquad \text{for all~$\u, \v \in \h^\C$}\:. \]
The spectral theorem for bounded symmetric operators on a complex Hilbert space yields the decomposition
\beq \label{DelNspec}
\Delta_N = \int_{\sigma(\Delta_N)} \lambda\: dE_\lambda \:,
\eeq
where~$dE_\lambda$ is a compactly supported projection-valued Borel measure on~$\R^+_0$
(for details see for example~\cite[Section~VII.3]{reed+simon}).
Let~$f$ be a real-valued (possibly unbounded) Borel function on~$\sigma(\Delta_N)
\subset \R^+_0$ which is finite almost everywhere.
Using the spectral calculus for (possibly unbounded) Borel functions
(for details see for example~\cite[Section~VIII.3]{reed+simon}), the operator
\beq \label{spectralcalc}
f\big(  \Delta_N \big) = \int_{\sigma(\Delta_N)} f(\lambda)\: dE_\lambda
\eeq
with domain
\[ \D\big( f\big(  \Delta_N \big) \big) := \Big\{ \u \in \h^\C \,\Big|\,
\int_{\sigma(\Delta_N)}  |f(\lambda)|^2 \: d\lla \u \,|\, E_\lambda \u \rra < \infty \Big\} \]
is selfadjoint. If~$E_{\lambda}=0$ for all points~$\lambda$ with~$f(\lambda) \in \{\pm \infty\}$,
then the domain is dense. In particular, this is the case if~$E$ has no point spectrum.
We remark that, if the function~$f$ is bounded, then the operator~$f(\Delta_N)$ is
also bounded, and its domain is the whole Hilbert space.

\begin{Remark} {\em{
We also have a spectral decomposition on the real Hilbert space $\h$ in the following sense.
Starting from the spectral decomposition~\ref{DelNspec}, we write the spectral measure
in a block operator form for the real and imaginary parts,
\beq \label{specmeas}
dE_\lambda = \begin{pmatrix}                             
		dE^{11}_\lambda & dE_\lambda^{12} \\[0.3em]
		dE_\lambda^{21} & dE_\lambda^{22} \end{pmatrix}\:.
\eeq
Next, by construction of the complexification, the operator~$\Delta_N$ has the block operator form
	\[ \Delta_N = \begin{pmatrix} \Delta_N^\R & 0 \\
	0 & \Delta_N^\R \end{pmatrix}\:, \]
where for clarity~$\Delta_N^\R : \h \rightarrow \h$ is the operator on the real Hilbert space.
Using the functional calculus, we find that for every real-valued Borel function $f$, the operator~$f(\Delta_N)$
is again block diagonal,
	\[ f(\Delta_N) = \begin{pmatrix} f(\Delta_N^\R) & 0 \\
		0 & f(\Delta_N^\R) \end{pmatrix}\:. \]
Comparing this equation with~\eqref{specmeas} and using that~$f$ is arbitrary,
we conclude that
\[ dE_\lambda^{11}=dE_\lambda^{22} \qquad \text{and} \qquad dE_\lambda^{12}=0=dE_\lambda^{21} \:. \]
Hence~$dE_\lambda^{11}$ is the desired
spectral measure of the operator~$\Delta_N^\R$. }} \QEDrem
\end{Remark}

\section{Construction of Weak Solutions} \label{secpoisson}
\subsection{The Weak Linearized Field Equations}
Given a suitable jet~$\w$, we want to solve the inhomogeneous linearized field equations~\eqref{Delvw}
in the weak formulation
\beq \label{weaksp}
\la \u, \Delta \v \ra_N = \la \u, \w \ra_N \qquad
\text{for all~$\u \in \Jvary$} \:,
\eeq
where we use the pointwise scalar product on the jets~\eqref{vsprod} in order to identify dual jets with jets.
Using~\eqref{DelNdef}, we can rewrite the left side with the operator~$\Delta_N$,
\beq \label{poissonprelim}
\lla \u, \Delta_N \v \rra = \la \u, \w \ra_N \qquad
\text{for all~$\u \in \Jvary$} \:.
\eeq
Using this formula in~\eqref{weaksp} has the disadvantage that we obtain different inner products on the left
and on the right.
This problem can be cured by absorbing one over the weight factor~\eqref{nudef} into the inhomogeneity.
Thus we write~\eqref{poissonprelim} equivalently as
\beq \label{weakpoisson}
\lla \u, \Delta_N \v \rra = \lla \u, \hat{\w} \rra \qquad
\text{for all~$\u \in \Jvary$} \:,
\eeq
where~$\hat{\w}$ is the new inhomogeneity
\[ \hat{\w}(\x) := \frac{1}{h(\x)}\: \w(\x) \:. \]
We refer to~\eqref{weakpoisson} as the {\em{weak linearized field equations}} formulated
with the adapted weighted $L^2$-scalar product.

\subsection{Solutions in Sobolev-Type Hilbert Spaces}
Formally, the weak linearized field equations~\eqref{weakpoisson} can be solved by
setting~$\v = (\Delta_N)^{-1} \u$.
Our strategy for making mathematical sense of the inverse is to work in suitable
Hilbert spaces constructed from the spectral calculus for the Laplacian~\eqref{spectralcalc}.
Since the Laplacian~$\Delta_N$ is bounded, the interesting point is the
behavior of the spectrum near zero. This leads us to consider negative powers of the spectral parameter.
\begin{Def} Given~$k \in \N_0$, the complex scalar product space~$(V, \lla .| . \rra_{\h^k})$
is defined by
\begin{align}
V &= \Big\{ \u \in \h^\C \:\Big|\: \int_{\sigma(\Delta_N)} \lambda^{-k} \: d\lla \u \,|\, E_\lambda \u \rra < \infty \Big\} 
\label{negk} \\
\lla \u | \v \rra_{\h^k} &= \int_{\sigma(\Delta_N)} \lambda^{-k} \: d\lla \u \,|\, E_\lambda \v \rra \:. \label{kneg}
\end{align}
Its Hilbert space completion is denoted by~$(\h^k, \lla .|. \rra_{\h^k})$.
We refer to the~$\h^k$ as {\bf{Sobolev-type Hilbert spaces}}.
\end{Def} \noindent
Clearly, the Hilbert space~$\h^0$ coincides with~$\h^\C$. If~$k \geq 1$, the condition for the integral in~\eqref{negk}
to be finite implies in particular that~$\u$ must be orthogonal to the kernel of~$\Delta_N$, i.e.\
\beq \label{orthokern}
\lla \u \,|\, \v \rra = 0 \qquad \text{for all~$\v \in \ker \Delta_N$} \:.
\eeq
Moreover, for any~$\varepsilon>0$, the subspace~$E_{(\varepsilon, \infty)}(\h^\C)$ is contained
in~$\h^k$. This means that~\eqref{kneg} poses a condition only for vectors whose spectral
decomposition extends to the origin.
The name {\em{Sobolev-type}} Hilbert space is motivated by the fact that
the spectral decomposition with respect to~$\Delta_N$ gives information on the
regularity of the jets; this will be explained in Section~\ref{secgaussian} in a simple example.

We also note that, the higher~$k$ is chosen, the stronger the condition in~\eqref{negk} becomes.
This shows that we have the sequence of inclusions
\[ \h^\C = \h^0 \supset \h^1 \supset \h^2 \supset \cdots \:.\]
Moreover, the Laplacian~$\Delta_N$ is a well-defined mapping
\[ \Delta_N : \h^k \rightarrow \h^{k+2} \:. \]

\begin{Thm} \label{thmmain} Let~$\hat{\w} \in \h^{k+2}$ with~$k \in \N_0$. Then the
inhomogeneous linearized field equations
\[ \Delta_N\, \u = \hat{\w} \]
have a unique weak solution~$\u \in \h^k$ given by
\beq \label{uint}
\u = \int_{\sigma(\Delta_N)} \lambda^{-1} \:d E_\lambda \hat{\w} \:,
\eeq
where the integral converges in~$\h^k$.
\end{Thm}
\Proof
First of all, we know from~\eqref{orthokern} that~$\hat{\w} \in (\ker \Delta_N)^\perp$.
Thus $E_{\{0\}}\hat{\w} =0$. As a consequence, the following integral exists and is finite,
\[ \int_{\sigma(\Delta_N)} \lambda^{-k}\: d\lla \u \,|\, E_\lambda \u \rra = \int_{\sigma(\Delta_N)} \lambda^{-k-2}\: d\lla \hat{\w} \,|\, E_\lambda \hat{\w} \rra = \|\hat{\w}\|^2_{h^{k+2}}< \infty \:. \]
As a consequence, the integral~\eqref{uint} converges in the Hilbert space~$\h^k$.
Moreover, it follows from the spectral calculus that
\[\Delta_N\u = \Delta_N \int_{\sigma(\Delta_N)} \lambda^{-1} \:d E_\lambda\: \hat{\w}
= \int_{\sigma(\Delta_N)} \lambda\, \lambda^{-1} \:d E_\lambda\: \hat{\w}
= \hat{\w}\:. \]
This concludes the proof.
\QED

Finally, it is convenient to introduce the function spaces
\[ \mathfrak{W}^k := \{ h(\x) \,\hat{\w}(\x) \:|\: \hat{\w} \in \h^k \} \:. \]
Then the mapping
\[ \Delta^{-1} : {\mathfrak{W}}^{k+2} \rightarrow \h^k \:,\qquad
\w \mapsto \v :=  \int_{\sigma(\Delta_N)} \lambda^{-1} \:d E_\lambda \:\frac{\w}{h} \]
gives the desired weak solution of the inhomogeneous linearized field equations~\eqref{Delvw}.

\section{A Few Simple Explicit Examples} \label{secexstatic}
In this section we illustrate the previous abstract constructions and results by a few
simple examples. These examples are chosen specifically in such a way that
a minimizing measure can be given in closed form,
making it possible to analyze the system explicitly.

\subsection{A One-Dimensional Gaussian} \label{secgaussian}
We let~$\G = \R$ and choose the Lagrangian as the Gaussian
\beq \label{gauss}
\L(\x,\y) = \frac{1}{\sqrt{\pi}}\: e^{-(\x-\y)^2} \:.
\eeq

\begin{Lemma} \label{lemmagauss} The Lebesgue measure
\[ d\mu = d\x \]
is a minimizer of the causal action principle for the Lagrangian~\eqref{gauss}
in the class of variations of finite volume (see~\eqref{Sdiff} and~\eqref{totvol}).
It is the unique minimizer within this class of variations.
\end{Lemma}
\Proof Writing the difference of the actions as in~\eqref{Sdiff}, we can carry out the integrals
over~$\mu$ using that the Gaussian is normalized,
\[ \int_\G \L(\x,\y)\: d\mu(\y) = 1 \:. \]
We thus obtain
\begin{align*}
\Sact(\mu) - \Sact(\tilde{\mu})
&= 2\, \int_N d(\mu-\tilde{\mu})(\x) +  \int_N d(\mu-\tilde{\mu})(\x) \int_N d(\mu-\tilde{\mu})(\y) \:\L(\x,\y) \\
&= \int_N d(\mu-\tilde{\mu})(\x) \int_N d(\mu-\tilde{\mu})(\y)\: \L(\x,\y) \:,
\end{align*}
where in the last line we used the volume constraint~\eqref{totvol}.
In order to show that the last double integral is positive, we take the Fourier transform and
use that the Fourier transform of a Gaussian is again a Gaussian. More precisely,
\[ \int_N e^{ip\x}\: \L(\x,\y) \:d\x = e^{-\frac{p^2}{4}} =: f(p) \:. \]
Moreover, the estimate
\[ \Big| \int_N e^{ip\x}\: d(\mu-\tilde{\mu})(\x) \Big| \leq \big| \tilde{\mu} - \mu \big|(\G) < \infty \]
shows that the Fourier transform of the signed measure~$\tilde{\mu} - \mu$ is 
a bounded function~$g \in L^\infty(\R)$.
Approximating this function in~$L^2(\R)$, we can
apply Plancherel's theorem and use the fact that convolution in position space
corresponds to multiplication in momentum space. We thus obtain
\begin{align}
\int_N &d(\mu-\tilde{\mu})(\x) \int_N d(\mu-\tilde{\mu})(\y)\: \L(\x,\y) \notag \\
&= \int_N \big( {\mathcal{F}}^{-1}(f g) \big)(\x) \: d(\mu-\tilde{\mu})(\x)
= \int_{-\infty}^\infty \overline{g(p)} \: e^{-\frac{p^2}{4}}\: g(p) \: dp \geq 0 \:, \label{gausspositive}
\end{align}
and the inequality is strict unless~$\tilde{\mu}=\mu$.
This concludes the proof.
\QED

The EL equations read
\beq \label{ELex1}
\int_\G \L(\x,\y) \: d\mu(\y) = 1 \qquad \text{for all~$\x \in \R$} \:.
\eeq
We now specify the jet spaces. Since the Lagrangian is smooth, it is obvious that
\[ \Jdiff = \J = C^\infty(\R) \oplus C^\infty(\R) \]
(where we identify a vector field~$a(\x)\: \partial_\x$ on~$\R$ with the function~$a(\x)$).
The choice of~$\Jtest$ is less obvious. In order to ensure that the conditions in
Definition~\ref{defJvary} are satisfied, we restrict attention to functions
which are bounded together with all their derivatives, denoted by
\[ C^\infty_\text{b} := \big\{ f \in C^\infty(\R) \:\big|\: f^{(n)} \in L^\infty(\R) \text{ for all~$n \in \N_0$} \big\} \:. \]
Now different choices are possible. Our first choice is to consider jets whose
scalar components are compactly supported,
\[ \Jtest = C^\infty_0(\R) \oplus C^\infty_\text{b}(\R) \:. \]
The linearized field equations~\eqref{eqlinlip} reduce to the scalar equation
\[ \int_N \big( \nabla_{1, \v} + \nabla_{2, \v} \big) \L(\x,\y)\: d\mu(\y) - \nabla_\v \,1 = 0 \qquad
\text{for all~$\x \in \R$} \:, \]
because if this equation holds, then the $\x$-derivative of the left side is also zero.
Using the EL equations~\eqref{ELex1}, the linearized field equations simplify to
\[ 
\int_N \nabla_{2,\v} \L(\x,\y)\: d\mu(\y) = 0 \qquad
\text{for all~$\x \in \R$}\:. \]

A specific class of solutions can be given explicitly. Indeed, choosing
\beq \label{inner}
\u = (a, A) \qquad \text{with} \qquad a \in C^\infty_0(\R) \text{ and } A(\x) := \int_{\infty}^\x a(t)\:dt \in C^\infty_{\text{b}}(\R) \:,
\eeq
integration by parts yields
\beq \label{pint}
\int_N \nabla_{2,\u} \L(\x,\y)\: d\mu(\y) = \int_N \big(A'(\y) + A(\y)\: \partial_\y \big) \L(\x,\y)\: d\y = 0 \:.
\eeq
These linearized solutions are referred to as {\em{inner solutions}},
as introduced in a more general context in~\cite{fockbosonic}. 
Inner solutions can be regarded as infinitesimal generators of transformations of~$N$ which leave the measure~$\mu$ unchanged. Therefore, inner solutions do not change the
causal fermion system, but merely describe symmetry transformations of the measure.
With this in mind, inner solutions are not of interest by themselves. But they can be used in order
to simplify the form of the jet spaces. We shall come back to these inner solutions in the general setting
in Appendix~\ref{appnoscalar}, where we will show that, by adding suitable inner solutions, one can
arrange that the test jets have vanishing scalar components. In our example, this can be arranged by
the transformation
\[ \v = (b,v) \mapsto \tilde{\v} := \v + \u \qquad \text{with} \qquad \u = (-b, -B) \:, \]
where~$B$ is an indefinite integral of~$b$.

In our example, we can also use the inner solutions alternatively in order to eliminate the vector component
of the test jets. To this end, it is preferable to choose the space of test jets as
\beq \label{veccompensate}
\Jtest = C^\infty_\text{b}(\R) \oplus C^\infty_\text{b}(\R) \:.
\eeq
Now the vector component disappears under the transformation
\[ \v = (b,v) \mapsto \tilde{\v} := \v + \u \qquad \text{with} \qquad \u = (-v', -v) \in \Jtest \:. \]
Therefore, it remains to consider the scalar components of jets. For technical simplicity, we
restrict attention to compactly supported functions. Thus we choose the jet space~$\Jvary$
in~\eqref{Jvarydef} as
\[ \Jvary = C^\infty_0(\R) \oplus \{0\} \:. \]
Then the Laplacian reduces to the integral operator with kernel~$\L(\x,\y)$,
\[ \big( \Delta (b,0) \big)(\x) = \int_\G \L(\x,\y)\: b(\y)\: d\y \:. \]
The bilinear form~\eqref{Delbilin} takes the form
\[ \la \,.\, , \Delta \,.\, \ra_N \::\: \Jvary \times \Jvary \rightarrow \R \:,\qquad
\la a, \Delta b \ra_N = \int_\G d\mu(\x) \int_\G d\mu(\y)\: \L(\x,\y)\: a(\x)\: b(\y) \:. \]
Moreover, $\| \overline{\nabla}^2 \L(\x,\y)\| = \L(\x,\y)$, so that the function~$h$ in~\eqref{nudef} is constant,
\[ h(\x) := 1 + \int_N \L(\x,\y) \: d\mu(\y) = 2 \:. \]
Therefore, the adapted weighted scalar product~$\lla ., \rra$ is twice the $L^2$-scalar product,
\[ \lla a, b \rra := 2 \int_N a b\: d\mu \:, \]
making it possible to identify~$\h$ with~$L^2(N, d\mu)$.
Consequently, the Laplacian~$\Delta_N$ defined by~\eqref{DelNdef} simply is the convolution with
the Gaussian,
\[ \Delta_N : \h \rightarrow \h \:, \qquad (\Delta_N a)(\x) = \frac{1}{2} \int_N \L(\x,\y)\, a(\y)\: d\mu(\y) \:. \]

The spectral decomposition of this operator is obtained by Fourier transformation,
\[ \Delta_N a = {\mathcal{F}}^{-1} \bigg( \frac{1}{2}\: e^{-\frac{p^2}{4}} \:\big( {\mathcal{F}} a \big)(p) \bigg) \:. \]
Thus, in momentum space, the operator~$\Delta_N$ is a multiplication operator.
Therefore, its spectral measure is obtained by multiplying with the characteristic function of the level sets,
\[ E_\Omega = {\mathcal{F}}^{-1} \, \chi_{g^{-1}(\Omega)} \,{\mathcal{F}} \qquad \text{with} \qquad g(p) := \frac{1}{2}\: e^{-\frac{p^2}{4}} \:. \]
In particular, one sees that the high momenta (i.e.\ large~$|p|$) correspond to the spectrum near zero.
Therefore, the negative powers of the spectral parameter in~\eqref{kneg} capture the
high-frequency behavior of the jets. This is quite similar to the usual Sobolev norms
on~$W^{2,k}(\R)$, as becomes clear when writing them in momentum space as
\[ \|\u\|^2_{W^{2,k}(\R)} = \int_{-\infty}^\infty \big(1+|p|^2 + \cdots + |p|^{2k}\big) \: |\hat{\u}(p)|^2\: dp \:. \]
This is why we refer to the~$\h^k$ as {\em{Sobolev-type Hilbert spaces}}.

Before going on, we remark that the above methods works more generally
if~$\L(\x,\y)$ is chosen as a function of~$\x-\y$ which has the properties that it is non-negative
and that its Fourier transform is strictly positive
(or non-negative, in which case the minimizing measure may not be unique).
In order to give a simple example, choosing the Lagrangian as the one-dimensional Yukawa potential,
\[ \L(\x,\y) = g(\x-\y) := \frac{e^{-|\x-\y|}}{|\x-\y|} \geq 0 \:, \]
its Fourier transform is computed by
\[\hat{g}(k) = \frac{4\pi}{1+|k|^2} \:,\]
which, as desired, is strictly positive. Therefore, the Lebesgue measure~$d\mu=d\x$ is 
again the unique minimizer of the causal action within the class of variations
of finite volume.

\subsection{A Minimizing Measure Supported on a Hyperplane} \label{secexhyper}
In the previous example, the support of the minimizing measure was the whole space~$\G$.
In most examples motivated from the physical applications,
however, the minimizing measure will be supported on a low-dimensional
subset of~$\G$ (see for instance the minimizers with singular support in~\cite{support, sphere}).
We now give a simple example where the minimizing measure is supported on a hyperplane of~$\G$.
We let~$\G=\R^2$ and choose the Lagrangian as
\beq \label{gauss2}
\L(\x,\y; \x',\y') = \frac{1}{\sqrt{\pi}}\: e^{-(\x-\x')^2} \big( 1 + \y^2 \big)\big( 1 +\y'^2 \big) \:,
\eeq
where~$(\x,\y), (\x',\y') \in \G$.
\begin{Lemma} \label{lemmagauss2} The measure
\beq \label{mingauss2}
d\mu = d\x \times \delta_\y
\eeq
(where~$\delta_\y$ is the Dirac measure)
is the unique minimizer of the causal action principle for the Lagrangian~\eqref{gauss2} under variations of
finite volume (see~\eqref{Sdiff} and~\eqref{totvol}).
\end{Lemma} \noindent
Note that this measure is supported on the $\x$-axis,
\[ N := \supp \mu = \R \times \{0\} \:. \]
\Proof[Proof of Lemma~\ref{lemmagauss2}]
Let~$\tilde{\mu}$ be a regular Borel measure on~$\G$
satisfying~\eqref{totvol}. Then the difference of actions~\eqref{Sdiff} is computed by
\begin{align}
&\Sact(\tilde{\mu}) - \Sact(\mu) 
= \frac{2}{\sqrt{\pi}} \int_\G d(\tilde{\mu}-\mu)(\x,\y) \int_N d\x' \:e^{-(\x-\x')^2} \:(1+ \y^2) \label{Slin} \\
&\quad\: +  \frac{1}{\sqrt{\pi}}
\int_\G d(\tilde{\mu}-\mu)(\x,\y) \int_\G d(\tilde{\mu}-\mu)(\x',\y') \:e^{-(\x-\x')^2}
\:\big( 1 + \y^2 \big)\big( 1 +\y'^2 \big) \:. \label{Squad}
\end{align}
Using that the negative part of the measure~$\tilde{\mu}-\mu$
is supported on the $\x$-axis, the first term~\eqref{Slin} can be estimated by
\begin{align*}
&\frac{2}{\sqrt{\pi}} \int_\G d(\tilde{\mu}-\mu)(\x,\y) \int_N d\x' \:e^{-(\x-\x')^2} \:(1+ \y^2) \\
&\overset{(*)}{\geq} \frac{2}{\sqrt{\pi}} \int_\G d(\tilde{\mu}-\mu)(\x,\y) \int_N d\x' \:e^{-(\x-\x')^2}
= \int_\G d(\tilde{\mu}-\mu)(\x,\y) = 0 \:,
\end{align*}
where in the last step we used the volume constraint.
The second term~\eqref{Squad}, on the other hand, can be rewritten as
\[ \frac{1}{\sqrt{\pi}} \int_\G d\rho(\x,\y) \int_\G d\rho(\x',\y') \:e^{-(\x-\x')^2} \]
with the signed measure~$\rho$ defined by
\[ d\rho(\x,\y) := \big( 1 + \y^2 \big)\: d(\tilde{\mu}-\mu)(\x,\y) \:. \]
Now we can proceed as in the proof of Lemma~\ref{lemmagauss} and use
that the Fourier transform of the integral kernel is strictly positive.
For the uniqueness statement one uses that the inequality in~$(*)$ is strict unless~$\tilde{\mu}$
is supported on the $\x$-axis. Then one can argue as in the proof of Lemma~\ref{lemmagauss}.
\QED

For the minimizing measure~\eqref{mingauss2}, the function~$\ell$ takes the form
\[ \ell(\x,\y) = \int_\G \L(\x,\y; \x',\y') \: d\mu(\x',\y') - 1 = \y^2 \:, \]
showing that the EL equations~\eqref{EL} are indeed satisfied.
We now specify the jet spaces. Since the Lagrangian is smooth, it is obvious that
\beq \label{Jdiffex}
\Jdiff = \J = C^\infty(\R) \oplus C^\infty(\R, \R^2) \:,
\eeq
where~$C^\infty(\R, \R^2)$ should be regarded as the space of two-dimensional
vector fields along the $\x$-axis.
Similar as explained after~\eqref{veccompensate}, we want to use the inner solutions for
simplifying the vector components of the jets. 
To this end, in analogy to~\eqref{veccompensate} we choose
\beq \label{veccompensate2}
\Jtest = C^\infty_\text{b}(\R) \oplus C^\infty_\text{b}(\R, \R^2) \:.
\eeq
The linearized field equations~\eqref{eqlinlip} read
\beq \label{lininter}
\nabla_{\u} \bigg( \int_{-\infty}^\infty \big( \nabla_{1, \v} + \nabla_{2, \v} \big) 
e^{-(\x-\x')^2} \big( 1 + \y^2 \big)\big( 1 +\y'^2 \big)
\: d\x' - \nabla_\v \,\sqrt{\pi} \bigg) \bigg|_{\y=\y'=0} = 0 \:.
\eeq
Now the inner solutions are generated by the vector fields tangential to the~$\x$-axis.
More precisely, in analogy to~\eqref{inner}, we consider the jet
\beq \label{inner2}
\v = \big(b, (B,0) \big) \qquad \text{with} \qquad b \in C^\infty_0 \text{ and } B(\x) := \int_{\infty}^\x b(t)\:dt \in C^\infty_{\text{b}}(\R) \:.
\eeq
Substituting this jet into~\eqref{lininter}, the linearized field equations simplify to
\[ \nabla_{\u} \bigg(\big( 1 + \y^2 \big) \int_{-\infty}^\infty \big( \nabla_{1, \v} + \nabla_{2, \v} \big) 
e^{-(\x-\x')^2} 
\: d\x' - \nabla_\v \,\sqrt{\pi} \bigg) \bigg|_{\y=\y'=0} = 0 \:. \]
The second component of the vector field~$u$ yields a $\y$-derivative,
giving rise to a factor~$2\y$, which vanishes at~$\y=0$. Therefore, it suffices to test with a vector field~$u$
which is tangential to the $\x$-axis. Now we are back in the example of the one-dimensional Gaussian.
Integrating by parts as in~\eqref{pint} one sees that the jet~$\v$ indeed satisfies the linearized field equations.

By suitably subtracting inner solutions, we can compensate the tangential components of the
jets. This leads us to choose
\beq \label{Jvaryex2}
\Jvary = C^\infty_0(\R) \oplus \big( \{0\} \oplus C^\infty_0(\R) \big) \:.
\eeq
Then the Laplacian simplifies as follows,
\begin{align*}
&\la \u, \Delta \v \ra(\x) \\
&= \frac{1}{\sqrt{\pi}}\: \nabla_{\u} \bigg( \int_{-\infty}^\infty \big( \nabla_{1, \v} + \nabla_{2, \v} \big) e^{-(\x-\x')^2} \:\big( 1 + \y^2 \big)\big( 1 +\y'^2 \big)\: d\x' - \nabla_\v \,\sqrt{\pi} \bigg) \bigg|_{\y=\y'=0} \\
&= \frac{2}{\sqrt{\pi}} \,u(\x)\, v(\x) \int_{-\infty}^\infty e^{-(\x-\x')^2}\: d\x' 
+ \frac{1}{\sqrt{\pi}}\: a(\x) \int_{-\infty}^\infty  e^{-(\x-\x')^2} \: b(\x')\: d\x' \\
&\quad\: + a(\x) \bigg( \frac{1}{\sqrt{\pi}} \int_{-\infty}^\infty b(\x)\: e^{-(\x-\x')^2} 
\: d\x' - b(\x) \bigg) \\
&= 2\,u(\x)\, v(\x) + \frac{1}{\sqrt{\pi}}\: a(\x) \int_{-\infty}^\infty  e^{-(\x-\x')^2} \: b(\x')\: d\x' \:,
\end{align*}
where~$\u=(a, (0,u))$ and~$\v=(b,(0,v))$. Hence the inhomogeneous linearized field equations~\eqref{Delvw} with $\w=(e,w)$
give rise to separate equations for the scalar and vector components,
\beq \label{linex2}
\frac{1}{\sqrt{\pi}} \int_{-\infty}^\infty  e^{-(\x-\x')^2} \: b(\x')\: d\x' = e(\x) \:,\qquad v(\x) =\frac{w(\x)}{2}\:.
\eeq

The bilinear form~\eqref{Delbilin} reduces to
\begin{align*}
\la \u, \Delta \v \ra_N = 2 \int_{-\infty}^\infty u(\x)\, v(\x)\: d\x + 
\frac{1}{\sqrt{\pi}} \int_{-\infty}^\infty d\x \int_{-\infty}^\infty  d\x'\: e^{-(\x-\x')^2} \: a(\x)\: b(\x') \:.
\end{align*}
Moreover, 
\begin{align*}
\overline{\nabla}_\u \overline{\nabla}_\v \L \big( \x,\y; \x',\y' \big) &= 
\frac{1}{\sqrt{\pi}} \: \big(a(\x) + a(\x') \big) \big(b(\x) + b(\x') \big) \: e^{-(\x-\x')^2} \\
&\quad\: + \frac{2}{\sqrt{\pi}} \,(u(\x)\, v(\x) + u(\x') v(\x') \big) \: e^{-(\x-\x')^2} \:.
\end{align*}
As a consequence, $\overline{\nabla}^2 \L$ is estimated from above and below by the Gaussian, i.e.\
\[ \frac{1}{c} \: e^{-(\x-\x')^2} \leq \|\overline{\nabla}^2 \L(\x,\y; \x',\y') \| \leq c\: e^{-(\x-\x')^2} \]
for a suitable numerical constant~$c>0$. As a consequence, the weight function~\eqref{nudef}
is bounded,
\[ h \big( (\x,0) \big) \leq 1+c \,\sqrt{\pi} \:. \]
Hence the adapted scalar product agrees, up to an irrelevant constant, with the $L^2$-scalar product.
This reflects the fact that the Laplacian is already bounded with respect to the $L^2$-scalar product,
making it unnecessary to introduce the adapted scalar product.

\subsection{An Adapted $L^2$-Scalar Product with Non-Trivial Weight} \label{secweightnontriv}
In the previous example, the weight function~$h(\x)$ in~\eqref{nudef} was constant
by symmetry.
Moreover, the linearized field equations for the vector component were local
(see the second equation in~\eqref{linex2}).
We now modify this example in order to make it more interesting in these respects.
We again let~$\G=\R^2$ and choose the Lagrangian as
\[ \L(\x,\y; \x',\y') = \frac{1}{\sqrt{\pi}}\: e^{-(\x-\x')^2} \: e^{2\y \y'}
\big( 1+ \x^2 \y^2 \big) \big( 1+ \x'^2 \y'^2 \big) \:. \]
In order to show that the measure~\eqref{mingauss2} is again a minimizer
of the causal action principle under variations of
of finite volume (see~\eqref{Sdiff} and~\eqref{totvol}), one argues as follows.
Similar as in the proof of Lemma~\ref{lemmagauss2}], the difference of
actions has a contribution linear in~$\tilde{\mu} - \mu$ (see~\eqref{Slin})
and a contribution quadratic in~$\tilde{\mu} - \mu$ (see~\eqref{Squad}).
Using again that the negative part of the measure~$\tilde{\mu}-\mu$
is supported on the $\x$-axis, the linear term is estimated by
\begin{align*}
& \int_\G d(\tilde{\mu}-\mu)(\x,\y) \int_G d\mu(\x',\y') \:\L(\x,\y; \x',\y') \\
&= \frac{1}{\sqrt{\pi}} \int_\G d(\tilde{\mu}-\mu)(\x,\y) \int_N d\x' \:
e^{-(\x-\x')^2} \:  \big( 1+ \x^2 \y^2 \big) \\
&\geq \frac{1}{\sqrt{\pi}} \int_\G d(\tilde{\mu}-\mu)(\x,\y) \int_N d\x' \:
e^{-(\x-\x')^2} = \int_\G d(\tilde{\mu}-\mu)(\x,\y) = 0 \:.
\end{align*}
The quadratic term in~$\tilde{\mu} - \mu$, on the other hand, can be rewritten as
\begin{align*}
&\int_\G d(\tilde{\mu}-\mu)(\x,\y) \int_\G d(\tilde{\mu}-\mu)(\x',\y') \:\L(\x,\y; \x',\y') \\
&= \frac{1}{\sqrt{\pi}} \int_\G  d\rho(\x,\y) \int_\G d\rho(\x',\y') \:
e^{-(\x-\x')^2}\: e^{-(\y-\y')^2} \:,
\end{align*}
where~$\rho$ is the signed measure
\[ d\rho(\x,\y) = \big( 1+ \x^2 \y^2 \big)\: e^{\y^2}\: d(\tilde{\mu}-\mu)(\x,\y) \:. \]
Now we can proceed again as in the proof of Lemma~\ref{lemmagauss} and use
that the Fourier transform of the integral kernel is strictly positive.

The function~$\ell$ takes the form
\begin{align*}
\ell(\x,\y) &= \int_\G \L(\x,\y; \x',\y') \: d\mu(\x',\y') - 1 \\
&= \frac{1}{\sqrt{\pi}} \int_{-\infty}^\infty  e^{-(\x-\x')^2} \:
\big( 1+ \x^2 \y^2 \big) \:d\x' - 1 = \x^2 \y^2 \:.
\end{align*}
We again choose the jet spaces~$\Jdiff$, $\Jtest$ and $\Jvary$ according to~\eqref{Jdiffex}, \eqref{veccompensate2}
and~\eqref{Jvaryex2}. Then
\begin{align*}
&\la \u, \Delta \v \ra(\x) \\
&= \nabla_{\u} \bigg( \frac{1}{\sqrt{\pi}}
\int_{-\infty}^\infty \big( \nabla_{1, \v} + \nabla_{2, \v} \big) 
e^{-(\x-\x')^2} \\
&\qquad\qquad\qquad\qquad \times \: e^{2\y \y'} \big( 1+ \x^2 \y^2 \big) \big( 1+ \x'^2 \y'^2 \big)
\: d\x' - \nabla_\v \,1 \bigg) \bigg|_{\y=\y'=0} \\
&= \frac{2}{\sqrt{\pi}} \,u(\x)\, v(\x) \int_{-\infty}^\infty e^{-(\x-\x')^2}\: \x^2 \:d\x' \\
&\quad\:-\frac{2}{\sqrt{\pi}} \,u(\x)\int_{-\infty}^\infty  e^{-(\x-\x')^2}\: v(\x') \:d\x' 
+ \frac{1}{\sqrt{\pi}}\: a(\x) \int_{-\infty}^\infty  e^{-(\x-\x')^2} \: b(\x')\: d\x' \\
&= 2\,u(\x)\, v(\x) \:\x^2 \\
&\quad\:-\frac{2}{\sqrt{\pi}} \,u(\x)\int_{-\infty}^\infty  e^{-(\x-\x')^2}\: v(\x') \:d\x'
+ \frac{1}{\sqrt{\pi}}\: a(\x) \int_{-\infty}^\infty  e^{-(\x-\x')^2} \: b(\x')\: d\x' \:.
\end{align*}
Thus, as desired, the linearized field equations are nonlocal also for the vector component.

Next, for any $(\x,\y),(\x',\y')\in \R\times \{0\}$
\begin{align*}
&\overline{\nabla}_\u \overline{\nabla}_\v \L \big( \x,\y; \x',\y' \big) = 
\frac{1}{\sqrt{\pi}} \: \big(a(\x) + a(\x') \big) \big(b(\x) + b(\x') \big) \: e^{-(\x-\x')^2} \\
&\quad\: + \frac{2}{\sqrt{\pi}} \,(u(\x)\, v(\x) \:\x^2+ u(\x') \,v(\x')\:\x'^2 \big) \: e^{-(\x-\x')^2} \:.
\end{align*}
Hence
\[ \| \overline{\nabla}_\u \overline{\nabla}_\v \L \big( \x,\y; \x',\y' \big) \| \simeq \big( 1 + \x^2 + \x'^2 \big)\:
e^{-(\x-\x')^2} \]
(here we make use of the fact that the jets in~$\Jvary$ can be chosen independently at
the points~$\x$ and~$\x'$). We conclude that the weight function~\eqref{nudef} in the adapted weighted
$L^2$-scalar product takes the form
\[ h(\x) = 1 + \int_N \big\| \overline{\nabla}^2 \L(\x,\y) \big\| \: d\mu(\y)
\simeq 1+ \x^2 \:. \]
Thus the adapted weighted $L^2$-scalar product is {\em{not}} equivalent to the standard $L^2$-scalar product.
This example explains why the adaptation of the weight is needed in order to realize the
Laplacian as a {\em{bounded}} symmetric operator on a Hilbert space.

\subsection{A Non-Homogeneous Minimizing Measure}
In the previous examples, the minimizing measure~$\mu$ was translation invariant in the
direction of the $\x$-axis. We now give a general procedure for constructing examples
of causal variational principles where the minimizing measure has no translational symmetry.
In order to work in a concrete example, our starting point is again the one-dimensional Gaussian~\eqref{gauss}.
But the method can be adapted to other kernels in a straightforward way.
In view of these generalizations, we begin with the following abstract result.

\begin{Lemma}
\label{MinCrit}
Let~$\rho$ be a measure on the $m$-dimensional manifold~$\G$ whose support is the whole manifold,
\[ \supp \rho = \G \:. \]
Moreover, let~$\L(\x,\y) \in (C^0 \cap L^\infty)(\G \times \G, \R^+_0)$ be a symmetric, non-negative, continuous and bounded kernel on~$\G \times \G$. Next, let~$h \in C^0(\G, \R^+)$ be a strictly positive, continuous function on~$\G$.
Assume that:
\bitem
\item[{\rm{(i)}}] $\displaystyle \int_\G \L(\x,\y) \: h(\y)\: d\rho(\y) = 1 \qquad \text{for all~$\x \in \G$}$. \\
\item[{\rm{(ii)}}] For all compactly supported bounded functions with zero mean,
\[ g \in L^\infty_0(\G, \R^+) \qquad \text{and} \qquad \int_\G g\: d\rho = 0 \:, \]
the following inequality holds,
\beq \label{geqineq}
\int_\G d\rho(\x) \int_\G d\rho(\y) \:\L(\x,\y)\: g(\x)\: g(\y) \geq 0 \:.
\eeq
\eitem
Then the measure~$d\mu := h\, d\rho$ is a minimizer of the causal action principle
under variations of finite volume (see~\eqref{Sdiff} and~\eqref{totvol}).
If the inequality~\eqref{geqineq} is strict for any non-zero~$g$, then the minimizing measure is unique
within the class of variations
\beq \label{measvar0}
d\tilde{\mu}_\tau= d\mu + \tau \,g \:d\rho \:.
\eeq
\end{Lemma}
\begin{proof} We begin with variations of the form~\eqref{measvar0}, which we write equivalently as
\beq \label{measvar}
d\tilde{\mu}_\tau = (h+\tau g)\: d\rho \:.
\eeq
Note that the function~$h$ is continuous and strictly positive.
Moreover, the function~$g$ is bounded and compactly supported.
This implies that the function~$h+\tau g$ is non-negative for sufficiently small~$|\tau|$.
Furthermore, using that~$g$ has mean zero, we conclude that~\eqref{measvar} is an admissible variation of finite volume~\eqref{totvol}. Moreover, the difference of the actions~\eqref{Sdiff} is well-defined and
computed by
\begin{align*}
&\Sact(\tilde{\rho}_\tau) - \Sact(\mu) \\
&= 2\tau \int_\G d\mu(\x)\:g(\x)\int_\G d\mu(\y) \:h(\y)\:\L(\x,\y)+ \tau^2 \int_\G d\mu(\x) \int_\G d\mu(\x) \:\L(\x,\y)\:g(\x)\:g(\y)\\ 
&\geq  2\tau \int_\G g(\y) \:d\mu(\y)= 0\;,
\end{align*}
where in the second step we used the above assumptions~(i) and (ii). The last step follows from the fact that~$g$
has mean zero.

We conclude that the measure~$\mu$ is a minimizer under variations of the form~\eqref{measvar}.
In order to treat a general variation of finite volume~\eqref{totvol}, we use the following approximation
argument. Our task is to show that
\beq \label{Nineq}
\int_\G d(\tilde{\mu}-\mu)(\x) \int_\G d(\tilde{\mu}-\mu) \:\L(\x,\y) \geq 0 \:.
\eeq
Exhausting~$\G$ by compact sets and using that the Lagrangian is bounded and that~$|\tilde{\mu}-\mu|$ is finite,
it suffices to consider the case that~$\tilde{\mu}-\mu$ is compactly supported,
\[ \supp (\tilde{\mu}-\mu) \subset K \Subset \G \:. \]
We choose a partition of unity~$(\eta_\ell)_{\ell \in \N}$ of~$\G$ which
is subordinate to the atlas of~$\G$. Then each measure~$\tilde{\mu}_\ell := \eta_\ell (\tilde{\mu}-\mu)$
is compactly supported in the domain of a chart. 
Thus we can identify it with a compactly
supported measure on~$\R^m$, which for ease in notation we again denote by~$\tilde{\mu}_\ell$.
Given~$k \in \N$, we decompose~$\R^m$ into cells of size~$1/k$,
\[ \R^m = \bigcup_{\vec{k} \in (\Z/k)^m} C_{\vec{k}} \qquad \text{with} \qquad
C_{\vec{k}} := \vec{k} + \Big[0, \frac{1}{k} \Big)^m \]
and define functions~$g^{(\ell)}_k$ by
\beq \label{gkdef}
g^{(\ell)}_k(\x) = \sum_{\vec{k} \in (\Z/k)^m} \frac{\tilde{\mu}_\ell(C_{\vec{k}})}{\rho(C_{\vec{k}})} \:\chi_{C_{\vec{k}}}(\x)
\eeq
(where~$\chi$ is again the characteristic function).
Note that in~\eqref{gkdef} only a finite number of summands are non-zero because the measure~$\tilde{\mu}_\ell$
is compactly supported. Moreover, each summand is finite because 
the functions~$g^{(\ell)}_k$ are bounded since the measure~$\rho(C_{\vec{k}})>0$; here we use
that the support of~$\rho$ is all of~$\G$.
Hence~$g^{(\ell)}_k \in L^\infty_0(\R^m)$, and a straightforward estimate shows that
\[ \lim_{k \rightarrow \infty} \int_{\R^m} \phi(\x)\: g^{(\ell)}_k(\x)\: d\rho(\x) = \int_{\R^m} \phi(\x)\: d\tilde{\mu}_\ell(\x) 
\qquad \text{for all~$\phi \in C^0_0(\R^m)$} \:. \]
In other words, the measures~$g^{(\ell)}_k\: d\rho$ converge in the weak$^*$-topology to the measure~$d\tilde{\mu}_\ell$.
Moreover, one verifies immediately that the approximation preserves the total volume
and the bound for the total variation, i.e.\
\[ \int_{\R^m} g^{(\ell)}_k\: d\rho = \tilde{\mu}_\ell\big( \R^m \big) \qquad \text{and} \qquad
\int_{\R^m} \big|g^{(\ell)}_k\big|\: d\rho \leq \big|\tilde{\mu}_\ell\big|
\:. \]
Finally, we introduce the function~$g_k$ by
\[ g_k(\x) = \sum_{\ell=1}^\infty g^{(\ell)}_k(\x) \:. \]
Here the sum is finite, because only a finite number of charts intersect~$K$.
Moreover, using that the measure~$\tilde{\mu}$ has finite total variation,
one sees from~\eqref{gkdef} that the functions~$g^\ell$ are all bounded. Moreover,
the measures~$g_k\: d\rho$ all have finite total variation.
Furthermore, they converge on~$K$ to the signed measure~$\tilde{\mu}-\mu$.
Using that the Lagrangian is continuous, it follows that
\begin{align*}
0 \leq &\int_N d\mu(\x) \int_N d\mu(\x) \:\L(\x,\y)\:g_k(\x)\:g_k(\y) \\
&\xrightarrow{k \rightarrow \infty}
\int_N d(\tilde{\mu}-\mu)(\x) \int_N d(\tilde{\mu}-\mu) \:\L(\x,\y) \:,
\end{align*}
concluding the proof.
\end{proof}

Our goal is to apply this lemma to kernels of the form
\beq \label{Lgaussdef}
\L(\x,\y) = f(\x) \: e^{-(\x-\y)^2} f(\y)
\eeq
with $\x,\y \in \R$ and a strictly positive function~$f$, which for convenience we again choose as a Gaussian,
\beq \label{fweight}
f(\x) = e^{\alpha \x^2} \qquad \text{with $\alpha \in \R$}\:.
\eeq
In order for this Lagrangian to be bounded, we choose~$\alpha<1$.
This kernel has the property~(ii) with respect to the Lebesgue-measure $d\rho=d\x$ because for all non-trivial~$g \in L^\infty_0(\G, \R^+)$,
\[ \int_\G d\rho(\x) \int_\G d\rho(\y) \:\L(\x,\y)\: g(\x)\: g(\y)
= \int_\G d\x \int_\G d\y \:e^{-(\x-\y)^2}\: (f g)(\x)\: (f g)(\y) > 0 \:, \]
where the last step is proved exactly as in the example of the Gaussian (see~\eqref{gausspositive})
In order to arrange~(i), for~$h$ we make an ansatz again with a Gaussian,
\beq \label{gaussweight}
h(\x) = c\: e^{\beta \x^2}\:.
\eeq
Then
\begin{align*}
&\int_\G \L(\x,\y) \: h(\y)\: d\rho(\y) = c \int_{-\infty}^\infty e^{\alpha \x^2} \: e^{-(\x-\y)^2}\: e^{(\alpha+\beta) \y^2}\: d\y \\
&= c \exp \Big( \alpha \x^2 - \x^2 - \frac{\x^2}{\alpha+\beta-1} \Big)
\int_{-\infty}^\infty \exp \bigg\{ (\alpha+\beta-1) \Big( \y + \frac{\x}{\alpha+\beta-1} \Big)^2 \bigg\}\:d\y \\
&= c\: \sqrt{\frac{\pi}{1-\alpha-\beta}}\: \exp \Big( \alpha \x^2 - \x^2 - \frac{\x^2}{\alpha+\beta-1} \Big) \:.
\end{align*}
In order to arrange that this function is constant one, we choose
\beq \label{cbetadef}
c = \sqrt{\frac{1-\alpha-\beta}{\pi}} \qquad \text{and} \qquad
\beta = -\frac{\alpha (2-\alpha)}{1-\alpha} \:.
\eeq
In order for the above Gaussian integral to converge, we need to ensure that~$1-\alpha-\beta > 0$.
In view of the formula
\[ 1 - \alpha - \beta = \frac{1}{1-\alpha} \:, \]
this inequality holds because we chose~$\alpha < 1$. Our finding is summarized as follows.
\begin{Prp} For any~$\alpha<1$, we let~$f$ and~$h$ be the Gaussians~\eqref{fweight} and~\eqref{gaussweight} with~$c$ and~$\beta$ according to~\eqref{cbetadef}. Then the measure~$d\mu=h\, d\x$ is a minimizer of the causal action corresponding to the Lagrangian~\eqref{Lgaussdef} within the class of variations of finite volume.
It is the unique minimizer within the class of variations~\eqref{measvar0}.
\end{Prp}

As a concrete example, we consider the well-known {\em{Mehler kernel}}
(see for example~\cite[Section~1.5]{glimm+jaffe})
\[ E(\x,\y) = \frac{1}{\sqrt{1-\rho^2}} \:\exp \bigg(-\frac{\rho^2 (\x^2+\y^2)- 2\rho \x \y}{(1-\rho^2)} \bigg) \]
with~$\rho > 0$. Rescaling~$\x$ and~$\y$ according to
\[ \x,\y \rightarrow \sqrt{\frac{1-\rho^2}{\rho}} \: \x, \sqrt{\frac{1-\rho^2}{\rho}} \:\y \:, \]
the Mehler kernel becomes
\[ E(\x,\y) = \frac{1}{\sqrt{1-\rho^2}} \:\exp \Big(-\rho (\x^2+\y^2)- 2 \x \y \Big) \:. \]
This kernel is of the desired form~\eqref{Lgaussdef} if we choose
\[ \alpha = 1-\rho < 1 \:,\qquad \beta = \frac{\rho^2-1}{\rho} \:.  \]

We finally remark that this non-homogeneous example can be used as the starting point
for the construction of higher-dimensional examples with minimizing measures supported on
lower-dimensional subsets, exactly as explained for the Gaussian in Section~\ref{secexhyper}.

\section{Example: Static Causal Fermion Systems} \label{seccfsstatic}
In view of the physical applications, the most important example of a causal variational principle
is the causal action principle for causal fermion systems. In this context, the methods developed
in the present paper should apply to {\em{static}} systems.
This can be understood in general terms as follows. As worked out in detail in~\cite{linhyp},
the linearized field equations can be analyzed with energy methods inspired from the
theory of hyperbolic PDEs. If one considers time-independent solutions, hyperbolic equations
(like the wave equation) give rise to corresponding elliptic equations (like the Poisson equation).
This suggests that, in the static situation, the linearized field equations should be analyzed with
methods from elliptic PDEs. As already mentioned in the introduction, this analogy
was our starting point for developing the methods in this paper.

Before introducing the static setting, we explain why the methods developed in the
present paper do {\em{not}} apply to the causal action principle in the {\em{time-dependent
setting}}. A simple explanation is that elliptic methods are not suitable for solving
hyperbolic equations. On a more technical level, the reasons are more involved.
Apart from regularity assumptions and technicalities, the main restriction for our
methods to apply is that the second variation of the Lagrangian~\eqref{D2L}
should be finite and integrable, in the sense that the weight function~\eqref{nudef} be finite
almost everywhere.
These conditions are harder to fulfill in the time-dependent setting, because
the norm in~\eqref{D2L} may be singular on the light cone.
Another problem is related to the kernel of the Laplacian~$\Delta_N$.
A hyperbolic equation (like the scalar wave equation) typically has a large kernel
(namely, all homogeneous equations like plane scalar waves).
In our treatment with Sobolev-like function spaces, we always assume that the
inhomogeneity is orthogonal to this kernel, and the constructed solution is also orthogonal
to this kernel. Such a treatment, even if it applies mathematically, does not seem
useful if the kernel of~$\Delta_N$ is too large. This is why we do not expect
the methods in this paper to be helpful for analyzing time-dependent situations.
Instead, it is preferable to use the hyperbolic methods developed in~\cite{linhyp}.

Static causal fermion systems were first considered in~\cite{pmt}. We here
present a somewhat simpler setting where the constraints are built in
right from the beginning.

\subsection{Causal Fermion Systems with Fixed Local Trace} \label{seccap}
We begin with the general definition of a causal fermion system with fixed local trace.
\begin{Def} \label{defcfs} {\em{ 
Given a separable complex Hilbert space~$\H$ with scalar product~$\la .|. \ra_\H$
and a parameter~$n \in \N$ (the {\em{spin dimension}}), we let~$\F \subset \Lin(\H)$ be the set of all
symmetric operators~$A$ on~$\H$ of finite rank which have trace one,
\beq \label{fixedtrace}
\tr A = 1 \:,
\eeq
and which (counting multiplicities) have
at most~$n$ positive and at most~$n$ negative eigenvalues. On~$\F$ we are given
a positive measure~$\rho$ (defined on a $\sigma$-algebra of subsets of~$\F$).
We refer to~$(\H, \F, \rho)$ as a {\em{causal fermion system with fixed local trace}}.
}}
\end{Def} \noindent
On~$\F$ we consider the topology induced by the operator norm
\[ 
\|A\| := \sup \big\{ \|A u \|_\H \text{ with } \| u \|_\H = 1 \big\} \:. \]
{\em{Spacetime}}~$M$ is defined to be the support of this measure,
\[ 
M := \supp \rho \subset \F \:. \]
It is a topological space (again with the topology induced by the operator norm).
The fact that the spacetime points are operators gives rise to
many additional structures which are {\em{inherent}} in the sense
that they only use information already encoded in the causal fermion system.
A detailed treatment can be found in~\cite[Section~1.1]{cfs}.

\subsection{The Reduced Causal Action Principle} \label{seccapreduced}
In order to single out the physically admissible
causal fermion systems, one must formulate physical equations. To this end, we impose that
the measure~$\rho$ should be a minimizer of the causal action principle,
which we now introduce. For any~$x, y \in \F$, the product~$x y$ is an operator of rank at most~$2n$. 
However, in general it is no longer symmetric because~$(xy)^* = yx$,
and this is different from~$xy$ unless~$x$ and~$y$ commute.
As a consequence, the eigenvalues of the operator~$xy$ are in general complex.
We denote the nontrivial eigenvalues counting algebraic multiplicities
by~$\lambda^{xy}_1, \ldots, \lambda^{xy}_{2n} \in \C$
(more specifically,
denoting the rank of~$xy$ by~$k \leq 2n$, we choose~$\lambda^{xy}_1, \ldots, \lambda^{xy}_{k}$ as all
the non-zero eigenvalues and set~$\lambda^{xy}_{k+1}, \ldots, \lambda^{xy}_{2n}=0$).
Given a parameter~$\kappa>0$ (which will be kept fixed throughout this paper),
we introduce the $\kappa$-Lagrangian and the causal action by
\begin{align*}
\text{$\kappa$-\em{Lagrangian:}} && \L(x,y) &= \frac{1}{4n} \sum_{i,j=1}^{2n} \Big( \big|\lambda^{xy}_i \big|
- \big|\lambda^{xy}_j \big| \Big)^2 + \kappa\: \bigg( \sum_{j=1}^{2n} \big|\lambda^{xy}_j \big| \bigg)^2 \\
\text{\em{causal action:}} && \Sact(\rho) &= \iint_{\F \times \F} \L(x,y)\: d\rho(x)\, d\rho(y) \:. 
\end{align*}
The {\em{reduced causal action principle}} is to minimize~$\Sact$ by varying the measure~$\rho$
under the
\[ \text{\em{volume constraint}} \qquad \rho(\F) = \text{const} \:. 
\]

This variational principle is obtained from the general causal action principle as introduced in~\cite[\S1.1.1]{cfs}
as follows. Using that minimizing measures are supported on operators of constant trace
(see~\cite[Proposition~1.4.1]{cfs}), we may fix the trace of the operators
and leave out the trace constraint. Moreover, 
by rescaling all the operators according to~$x \rightarrow \lambda x$ with~$\lambda \in \R$,
we may assume without loss of generality that this trace is equal to one~\eqref{fixedtrace}.
Next, the $\kappa$-Lagrangian arises when treating the so-called boundedness constraint
with a Lagrange multiplier term. Here we slightly simplified the setting
by combining this Lagrange multiplier term with the Lagrangian.

\subsection{Static Causal Fermion Systems} \label{secstatic}
We now specialize our setting to the static case.
Adapting the causal action principle to static causal fermion systems and imposing a
regularity condition, we will then get into the setting of causal variational principles (Sections~\ref{secSstatic}
and~\ref{secregular}).

\begin{Def} \label{defstatic}
Let~$(\scrU_t)_{t \in \R}$ be a strongly continuous one-parameter group of unitary transformations 
on the Hilbert space~$\H$ (i.e.\ $s$-$\lim_{t' \rightarrow t} \scrU_{t'}= \scrU_t$ and~$\scrU_t \scrU_{t'} = \scrU_{t+t'}$).
The causal fermion system~$(\H, \F, \rho)$ is {\bf{static with respect to~$(\scrU_t)_{t \in \R}$}}
if it has the following properties:
\begin{itemize}[leftmargin=2em]
\item[\rm{(i)}] Space-time $M:= \supp \rho \subset \F$ is a topological product,
\beq \label{inffin}
M = \R \times N \qquad \text{or} \qquad M = S^1 \times N \:.
\eeq
We write a space-time point~$x \in M$ as~$x=(t,\x)$ with~$t \in \R$ and~$\x \in N$,
where in the case~$M=S^1 \times N$ we identify the circle with the unit interval via
\[ S^1 \simeq \R \text{ mod } \Z \:. \]
\item[\rm{(ii)}] The one-parameter group~$(\scrU_t)_{t \in \R}$ leaves
the measure~$\rho$ invariant, i.e.
\[ \rho\big( \scrU_t \,\Omega\, \scrU_t^{-1} \big) = \rho(\Omega) \qquad \text{for all
	$\rho$-measurable~$\Omega \subset \F$} \:. \]
Moreover,
\[ \scrU_{t'}\: (t,\x)\: \scrU_{t'}^{-1} = (t+t',\x) \]
(where in the case~$M=S^1 \times N$ the sum~$t+t'$ is taken modulo~$\Z$).
\end{itemize}
\end{Def} \noindent
The two cases in~\eqref{inffin} are referred to as spacetimes
of {\em{finite}} and {\em{infinite lifetime}}, respectively. Here we can treat both cases together.

Given a static causal fermion system, we also consider the set of operators
\[ N := \{ (0, \x) \} \subset \F \:. \]
The measure~$\rho$ induces a measure~$\mu$ on~$N$ defined by
\[ \mu(\Omega) := \rho\big( [0,1] \times \Omega \big) \:. \]
The fact that the causal fermion system is static implies that~$\rho([t_1,t_2] \times \Omega) = (t_2-t_1)\,
\mu(\Omega)$, valid for all~$t_1<t_2$. This can be expressed more conveniently as
\[ 
d\rho = dt \,d\mu \:. \]

\subsection{The Causal Action Principle in the Static Setting} \label{secSstatic}
The causal action principle can be formulated in a straightforward manner
for static causal fermion systems. The only point to keep in mind is that, when considering
families of measures, these measures should all be static with respect to the same
group~$(\scrU_t)_{t \in \R}$ of unitary operators (see Definition~\ref{defstatic}).
In order to make this point clear, right from the beginning we
choose a group of unitary operators~$(\scrU_t)_{t \in \R}$ on~$\H$.
We denote the equivalence classes of~$\F$ under the action of the one-parameter group by
\[ \G:= \F/\R := \{ \scrU_t \,x\, \scrU_t^{-1} \:|\: x \in \F, t \in \R \} \:. \]
We denote the elements of~$\F/\R$ just as the spatial points by~$\x$ and~$\y$.
Next, we define the following functions:
\begin{align}
\text{\em{static $\kappa$-Lagrangian}} &&
\L(\x, \y) &:= \int_I \L \big((0,\x), (t,\y) \big) \: dt \:, \label{Ltime}
\end{align}
where~$I$ is a generalized interval chosen in infinite lifetime as~$I=\R$,
and in finite lifetime as~$I=[0,1)$.
Note that, for ease in notation, we use the same symbol for the static as for the original Lagrangian. But they can be distinguished by their arguments, because the static Lagrangian depends on spatial points in~$\G$ whereas the original Lagrangian is defined on space-time operators in~$\F$. 
For a measure~$\rho$ which is static with respect to~$(\scrU_t)_{t \in \R}$, we introduce the
\begin{align*}
\text{\em{static causal action}} && \Sact(\mu) &= \int_{\F/\R} d\mu(\x) \int_{\F/\R} d\mu(\y) \:\L(\x, \y) \:.
\end{align*}
The {\em{static causal action principle}} is to minimize~$\Sact$ by varying the measure~$\mu$
within the class of regular Borel measures on~$\F/\R$ under the
\begin{align*}
\text{\em{volume constraint}} && \mu(\F/\R) = \text{const} \:.
\end{align*}

\subsection{The Regular Setting as a Causal Variational Principle} \label{secregular}
We now explain how to get to the setting of causal variational principles introduced
in Section~\ref{seccvp}. In order to give the set of operators a manifold structure,
we assume that~$\rho$ is {\em{regular}} in the sense that all operators in its support
have exactly~$n$ positive and exactly~$n$ negative eigenvalues.
This leads us to introduce the set~$\F^\text{reg}$ 
as the set of all linear operators~$F$ on~$\H$ with the following properties:
\begin{itemize}[leftmargin=2em]
\item[(i)] $F$ is selfadjoint, has finite rank and (counting multiplicities) has
exactly~$n$ positive and~$n$ negative eigenvalues. \\[-0.8em]
\item[(ii)] The trace is constant, i.e.~$\tr(F) = c>0$ (with $c$ independent of $F \in \F^\text{reg}$).
\end{itemize}

At this point, one must distinguish the cases that~$\H$ is finite- or infinite-dimensional.
In the {\em{infinite-dimensional setting}}, the set~$\F^\text{reg}$ is an infinite-dimensional
Banach manifold (for details see~\cite{banach}).
In order to get into the setting of causal variational principles with a locally compact manifold~$\G$,
one must restrict attention to a finite-dimensional submanifold of~$\F^\text{reg}$.
Clearly, this submanifold must contain the supports of both measures~$\rho$ and all its considered variations,
and the unitary group~$(\scrU_t)_{t \in \R}$ must map the submanifold to itself.
Moreover, the vector fields of the jets needed for the analysis must all be
tangential to this submanifold. Then we can simply choose~$\G$ as the equivalence classes of this submanifold
under the action of the group~$(\scrU_t)_{t \in \R}$.
One also needs to verify that the resulting static $\kappa$-Lagrangian~\eqref{Ltime}
has the properties~(i)--(iv) in Section~\ref{seccvp}.

The {\em{finite-dimensional setting}} is considerably easier.
In this case, the set~$\F^\text{reg}$ has a smooth manifold structure
(see the concept of a flag manifold in~\cite{helgason} or the detailed construction
in~\cite[Section~3]{gaugefix}).
Assuming that the action of the group~$(\scrU_t)_{t \in \R}$ on~$\F^\text{reg}$ is proper and has no fixed points,
the quotient
\[ \G := \F^\text{reg} / \R \]
is again a manifold. In this way,
we get into the setting of causal variational principles as introduced in Section~\ref{seccvp}.
The resulting static $\kappa$-Lagrangian indeed has the properties~(i)--(iv) in Section~\ref{seccvp}
(for details see~\cite[Section~3.3]{pmt}).

\section{Conclusion and Outlook} \label{secoutlook}
In the present paper we showed that linearized fields of causal variational principles
can be analyzed with functional analytic tools in suitable function spaces.
This opens the door to the analysis of corresponding nonlinear equations, in particular of the
restricted EL equations~\eqref{ELtest}.
For the explicit analysis one can proceed perturbatively as worked out in general in~\cite{perturb}.
But one can also use and adapt methods of nonlinear elliptic partial differential equations.
In particular, it seems a promising strategy to
analyze nonlinear equations with fixed point methods and suitable a-priori estimates.

\appendix
\section{The Kernel of the Laplacian} \label{appnoscalar}
The kernel of the Laplacian~$\Delta_N$ plays an important role for the solvability of the
linearized field equations~\eqref{weakpoisson}. Namely, since the operator~$\Delta_N$ is symmetric, its image is the orthogonal
complement of its kernel. Therefore, one sees immediately from~\eqref{weakpoisson}
that the weak linearized field equations do not
admit solutions unless the inhomogeneity~$\hat{\w}$ is in
the orthogonal complement of the kernel of~$\Delta_N$. This is reflected in our general existence
results by the fact that the vectors in the weighted Hilbert spaces~$\h^k$ for~$k>0$ are in the orthogonal
complement of the kernel of the operator~$\Delta_N$ (see~\eqref{orthokern}).
With this in mind, it is an important task to analyze this kernel.

In this appendix we show that, under additional smoothness assumptions, there is indeed an explicit
class of jets which are in the kernel. We also show how these jets can be treated
when constructing inhomogeneous solutions. We need to make the following assumptions.
\begin{Def} \label{defsms}
The support~$N:= \supp \mu$ has a {\bf{smooth manifold structure}} if
the following conditions hold:
\bitem
\item[\rm{(i)}] $N$ is a $k$-dimensional  smooth, oriented and connected submanifold of~$\F$.
Equipped with a smooth atlas, we also denote it by~$\scrN$.
\item[\rm{(ii)}] In a chart~$(x,U)$ of~$\scrN$, the measure~$\mu$ is absolutely continuous with respect
to the Lebesgue measure with a smooth, strictly positive weight function,
\[ 
d\mu = h(x)\: d^kx \qquad \text{with} \quad h \in C^\infty(\scrN, \R^+) \:. \]
\eitem
\end{Def} 

Let~$v \in \Gamma(N, TN)$ be a vector field. Then
its {\em{divergence}} $\div v \in C^\infty(N, \R)$ may be defined by the relation
\[ \int_N \div v\: \eta(x)\: d\mu = -\int_N D_v \eta(x)\: d\mu(x) \:, \]
to be satisfied by all test functions~$\eta \in C^\infty_0(N, \R)$.
In a local chart~$(x,U)$, the divergence is computed by
\[ 
\div v = \frac{1}{h}\: \partial_j \big( h\, v^j \big) \]
(where, following the Einstein summation convention, we sum over~$j=1,\ldots,k$).

When integrating by parts using the Gau{\ss} divergence theorem, we need to make
sure that we do not get boundary values at infinity. To this end, it is convenient
to choose the Riemannian metric~$g_x$ introduced before~\eqref{vsprod}
to be compatible with the smooth manifold structure in the following sense.

\begin{Def} \label{defgadapted} The Riemannian metric~$g$ on~$\Gamma_x$ is {\bf{adapted at infinity}} if
there is a sequence~$(\eta_n)_{n \in \N}$ of compactly supported functions,
$\eta_n \in C^\infty_0(N, \R)$, with the following properties:
\bitem
\item[{\rm{(i)}}] The functions~$\eta_n$ are non-negative, monotone increasing and exhaust~$N$
in the sense that for any compact set~$K \subset N$ there is~$N$ with~$\eta_n|_K \equiv 1$
for all~$n \geq N$.
\item[{\rm{(ii)}}] The derivatives tend uniformly to zero, i.e.
\[ \lim_{n \rightarrow \infty} \sup_{x \in N} \|D \eta_n\|_x = 0 \:, \]
where~$\|.\|_x$ is the norm on~$T_xN \subset \Gamma_x$ induced by the Riemannian metric~$g$.
\eitem
\end{Def} \noindent
This definition poses an implicit condition on the behavior of the Riemannian metric~$g$
at infinity. A typical example is that~$N$ has one asymptotic end in the sense
that there is a compact subset~$K \subset N$ such that~$N \setminus K$ is diffeomorphic
to~$\R^k$ minus a closed ball. In this case, one can choose the Riemannian metric~$g$
outside~$K$ as the pull-back of the Euclidean metric, and extend it smoothly to~$N$.
Then one can choose the~$\eta_n$
as the pull-back of cutoff functions in~$\R^k$ which fall off for example in annular regions~$B_{2n} \setminus B_n$.

\begin{Def} \label{definner} An {\bf{inner solution}} is a jet~$\v$ of the form
\[ \v = (\div v, v) \qquad \text{with} \qquad v \in \Gamma(N, TN) \:. \]
We make the following regularity and decay assumptions:
\bitem
\item[{\rm{(i)}}] The vector field~$v$
can be extended to a vector field~$\tilde{v} \in \Gamma(U, T\F)$ defined in a neighborhood~$U$ of~$N$
such that the directional derivative~$(D_{1,\tilde{v}} + D_{2,\tilde{v}}) \L(x,y)$
exists for all~$x \in U$ and~$y \in N$ and is integrable in~$y$, i.e.\
\[ \int_N \Big| \big(D_{1,\tilde{v}} + D_{2,\tilde{v}} \big) \L(x,y) \Big| \: d\mu(y) < \infty \qquad
\text{for all~$x \in U$}\:. \]
Moreover, the directional derivative~$D_{\tilde{v}} \ell(x)$ exists for all~$x \in U$ and is continuous in~$U$.
\item[{\rm{(ii)}}] The integral
\[ \int_N \L(x,y)\: \|\v(y)\|_y \:d\mu(y) \]
is finite and bounded locally uniformly in a neighborhood of~$N$
(where~$\|.\|_y$ is again the norm corresponding to the scalar product~\eqref{vsprod}
and the Riemannian metric is adapted at infinity according to Definition~\ref{defgadapted}).
\item[{\rm{(iii)}}] For any test jet~$\u \in \Jtest_\mu$, the
directional derivative~$D_v \u$ (computed in the same charts used for computing the higher derivatives in Definition~\ref{defJvary}) is again in~$\Jtest_\mu$.
\eitem
The vector space of all inner solutions is denoted by~$\Jin_\mu$.
\end{Def} \noindent
Note that~(i) implies that every inner solution is in~$\J^1_\mu \cap \Jdiff_\mu$
(see~\eqref{Jdiffdef} and Definition~\ref{defJvary}).

The name ``inner {\em{solution}}'' is justified by the following lemma:
\begin{Lemma} Let~$\mu$ be a critical measure (meaning that the
EL equations~\eqref{EL} hold).
Then every inner solution~$\v \in \Jin_\mu$ is a solution of the linearized field equations, i.e.\
\[ \la \u, \Delta \v \ra|_N = 0 \qquad \text{for all~$\u \in \Jtest_\mu$} \:. \]
\end{Lemma} \noindent
The proof is given in~\cite[Lemma~3.3]{fockbosonic}.

We next show that for any function~$a$ on~$N$ one can find an inner solution
whose scalar component coincides with~$a$. If~$N$ were compact, the analogous
statement would be the infinitesimal version of Moser's theorem (see for
example~\cite[Section~XVIII, \S2]{langDG}).
Here we give a detailed proof if~$N$ is non-compact,
based on~\cite[Theorem~1.2 in Section~XVIII]{langDG}. 
\begin{Prp} \label{prpmoser}
Assume that~$N$ has a smooth manifold structure
and that the Riemannian metric~$g$ on~$\Gamma_x$ is adapted at infinity
(see Definitions~\ref{defsms} and~\ref{defgadapted}).  Then to any $a\in C^{\infty}(N)$ we can find a vector field $u\in \Gamma(N,TN)$ with $\div u=a$.
\end{Prp}
\begin{proof} We again consider the functions~$\eta_n$ in Definition~\ref{defgadapted}.
By construction, these functions~$\eta_n$ are monotone increasing and exhaust~$N$
in the sense that for any compact set~$K \subset N$ there is~$N$ with~$\eta_n|_K \equiv 1$
for all~$n \geq N$. Next, we introduce the non-negative functions
\[ \phi_n:= \left\{ \begin{array}{cl}
	\eta_1 \quad &\text{if~$n=1$} \\[0.1em]
	\eta_{n}-\eta_{n-1} \quad & \text{if~$n>1$} \:. \end{array} \right. \]
Possibly by leaving out some of these functions we can arrange that none of the~$\phi_n$
is identically equal to zero. Then all $\phi_n$ are in $C^\infty_0(N,\R)$ and form a locally finite partition of unity. Moreover, the sets~$N_n$ defined by
\[ N_n:= \left\{ \begin{array}{cl}
	\supp\phi_1 \quad &\text{if~$n=1$} \\[0.1em]
	\supp\phi_n \cup \supp\phi_{n-1} \quad & \text{if~$n>1$} \:.
\end{array} \right. \]
form a locally finite covering of~$N$.

We adapt the proof of~\cite[Lemma 2.7]{pmt} to our setting. We
express the measure $\mu$ using a volume form $\psi \in \Lambda^k(N)$, i.e.:
\begin{flalign*}
	\mu(U) = \int_U \psi \qquad \mathrm{for\;all\;compact\;}U\subseteq N\;.
\end{flalign*}
Similarly, we can also find a volume form $\omega \in \Lambda^k(N)$ representing $a\mu$:
\begin{flalign*}
	\int_U a(x) \:d\mu(x) = \int_U \omega \;,\;\;\;\mathrm{for\;all\;compact\;}U\subseteq N\;.
\end{flalign*}
Now choose $c_1 \in \R$ such that:
\begin{flalign*}
	0=\int_N(\phi_1 \,\omega - c_1\, \phi_1 \,\psi)= \int_{N_1}(\phi_1 \,\omega - c_1\, \phi_1\, \psi)\;.
\end{flalign*}
Then, in view of~\cite[§ XVIII, Theorem 1.2]{langDG} applied to $N_1$ (which, due to our assumptions, is a connected oriented manifold without boundary), there exists a compactly supported $(k-1)$-form $\eta_1 \in \Lambda_0^{(k-1)}(N_1)$ with
\begin{flalign*}
	(\phi_1 \,\omega - c_1 \,\phi_1 \,\psi)\big|_{N_1} = d\eta_1\:.
\end{flalign*}
Clearly, we can extended~$\eta_1$ by zero to~$N$.

Now for any $n\in \N$, inductively choose $c_{n+1} \in \R$ such that
\begin{flalign*}
\int_N \Big( \phi_{n+1} \,\omega + c_n \,\phi_n \,\psi - c_{n+1}\, \phi_{n+1}\psi \Big)
= \int_{N_{n+1}} \Big( \phi_{n+1} \,\omega + c_n \, \phi_n\, \psi - c_{n+1}\, \phi_{n+1}\psi \Big)= 0\:.
\end{flalign*}
Now~\cite[§ XVIII, Theorem 1.2]{langDG} again provides a $(k-1)$-form $\eta_{n+1}\in \Lambda^{(k-1)}_0(N_{n+1})$ such that
\begin{flalign*}
	\phi_{n+1} \, \omega + c_n\, \phi_n\, \psi - c_{n+1}\, \phi_{n+1}\psi = d\eta_{n+1}\:,
\end{flalign*}
where we again extended~$\eta_{n+1}$ by zero to~$N$. For ease in notation, from now on by $\eta_{n}$ we always mean the smooth  extensions to~$N$. Then each~$\eta_n$ has compact support in $N_n$.

As $(N_n)_{n\in \N}$ is locally finite, the series~$\eta:= \sum_{n=1}^{\infty}\eta_n$ converges
and can be computed to be
\begin{flalign*}
d\eta &= \sum_{n=1}^{\infty}d\eta_n = \phi_1\, \omega - c_1\, \phi_1\, \psi
+ \sum_{n=1}^{\infty}\Big( \phi_{n+1}\, \omega + c_n\, \phi_n\, \psi-c_{n+1}\, \phi_{n+1}\, \psi \Big) \\
&= - c_1\, \phi_1\, \psi +\omega \sum_{n=1}^{\infty} \phi_{n}+ \sum_{n=1}^{\infty} \Big( c_n\, \phi_n\, \psi
- c_{n+1}\, \phi_{n+1}\, \psi\Big)\:.
\end{flalign*}
The last sum is telescopic, giving $c_1\phi_1\psi$, which cancels the first summand in the last line. Furthermore, as $(\phi_n)_{n\in\N}$ is a partition of unity, the second summand in the last line reduces to $\omega$.
We conclude that	\begin{flalign*}
	d\eta = \omega\;.
\end{flalign*}
It remains to find a vector field $u\in \Gamma(N,TN)$ such that $\div u d\mu = d\eta$. This can be done just as described in~\cite[proof of Lemma~2.7]{pmt}: By a conformal transformation of the metric~$g$
we arrange that the corresponding volume form coincides with~$\mu$. Then we set
\[ u^{\alpha} := \hat{g}^{\alpha \beta}(\star \eta)_{\beta} \]
with the Hodge star $\star: \Lambda^{(k-1)}(N) \rightarrow \Lambda^1(N)$.  Then $u$ has the desired property (up to a sign depending on the dimension of $N$).
\end{proof}

Let us discuss the significance of the inner solutions. As already mentioned in
Section~\ref{secgaussian} in a simple example, inner solutions can be regarded as
infinitesimal generators of transformations of~$N$ which leave the measure~$\mu$ unchanged.
Similar to gauge transformations, they describe infinitesimal symmetries
of the system. Using these symmetry transformations, one can simplify the form of the jet spaces.
For example, under suitable assumptions
we can remove the scalar components of linearized solutions, as we now explain.
Let~$\v=(b,v) \in \Jvary$. Then, according to~\eqref{Jvarydef}, the scalar component~$b$
has compact support. We choose~$a=-b$. Applying Proposition~\ref{prpmoser}, there
is a smooth vector field~$u$ with~$\div u = a$. If~$a$ has mean zero,
\[ \int_N a(\x)\: d\mu(\x) = 0 \:, \]
then~$u$ can be chosen to be also compactly supported
(as one sees immediately from the above proof or alternatively in~\cite[§ XVIII, Theorem 1.2]{langDG}).
Otherwise, the vector field~$u$ typically decays at infinity, 
but is not compactly supported.
As a simple example, in the case of one asymptotic end introduced after Definition~\ref{defgadapted},
the vector field~$h u$ is divergence-free in the asymptotic end with respect to the Euclidean metric.
Applying the Gauss divergence theorem, the flux of this vector field through a sphere of radius~$r$
does not depend on $r$, making it possible to arrange the decay rate
\[ (h u)(\x) \lesssim \frac{1}{\|x\|^{k-1}} \]
(where~$\|.\|$ is the Euclidean norm).
As a consequence, the integral in~\eqref{normadapt} is finite if~$k \geq 3$ and~$d\mu$
goes over asymptotically to the Lebesgue measure.
More generally, here we assume for simplicity that the resulting jet~$\u := (a,u)$ has a finite adapted norm,
$\norm \u \norm<\infty$. Then we can redefine~$\Jvary$ as all the jets~$\v+\u$.
In contrast to~\eqref{Jvarydef}, the vector components of the resulting jets are no longer
compactly supported, but the adapted norm is still finite,
making sure that our existence proof
still goes through. As a result of this procedure, the scalar components of all linearized solutions vanish.
As a further advantage, the kernel of the operator~$\Delta_N$ has become smaller.

We finally point out that, even after the above procedure, in typical examples the Laplacian~$\Delta_N$
still has a non-trivial kernel. For example, every compactly supported
divergence-free vector field~$v$ gives rise to an inner solution~$(0,v)$ with
vanishing scalar component, which lies in the kernel of the Laplacian.
In order to mod out the corresponding infinitesimal symmetry transformations
in the construction of inhomogeneous linearized field equation, 
we set up the Hilbert space formulation in such a way that the inhomogeneity~$\hat{\w}$ must be in the orthogonal complement of the kernel of~$\Delta_N$. Likewise, the constructed solution is orthogonal to the kernel. In this way,
we can solve the linearized field equations modulo vectors in the kernel of the Laplacian.

\Thanks{{{\em{Acknowledgments:}}
We are grateful to the referee for helpful comments on the manuscript.
M.L.\ gratefully acknowledges support by the Studienstiftung des deutschen Volkes and the Marianne-Plehn-Programm.


\begin{thebibliography}{10}

\bibitem{cfsweblink}
\emph{Link to web platform on causal fermion systems:
  \href{https://www.causal-fermion-system.com}{www.causal-fermion-system.com}}.

%
\bibitem{sphere}
L.~B\"auml, F.~Finster, H.~von~der Mosel, and D.~Schiefeneder, \emph{Singular
  support of minimizers of the causal variational principle on the sphere},
  arXiv:1808.09754 [math.CA], Calc. Var. Partial Differential Equations
  \textbf{58} (2019), no.~6, 205.

\bibitem{bogachev}
V.I. Bogachev, \emph{Measure {T}heory. {V}ol. {I}}, Springer-Verlag, Berlin,
  2007.

\bibitem{linhyp}
C.~Dappiaggi and F.~Finster, \emph{Linearized fields for causal variational
  principles: {E}xistence theory and causal structure}, arXiv:1811.10587
  [math-ph], Methods Appl. Anal. \textbf{27} (2020), no.~1, 1--56.

\bibitem{continuum}
F.~Finster, \emph{Causal variational principles on measure spaces},
  arXiv:0811.2666 [math-ph], J. Reine Angew. Math. \textbf{646} (2010),
  141--194.

\bibitem{cfs}
\bysame, \emph{The {C}ontinuum {L}imit of {C}ausal {F}ermion {S}ystems},
  arXiv:1605.04742 [math-ph], Fundamental Theories of Physics, vol. 186,
  Springer, 2016.

\bibitem{nrstg}
\bysame, \emph{Causal fermion systems: {A} primer for {L}orentzian geometers},
  arXiv:1709.04781 [math-ph], J. Phys.: Conf. Ser. \textbf{968} (2018), 012004.

\bibitem{positive}
\bysame, \emph{Positive functionals induced by minimizers of causal variational
  principles}, arXiv:1708.07817 [math-ph], Vietnam J. Math. \textbf{47} (2019),
  23--37.

\bibitem{perturb}
\bysame, \emph{Perturbation theory for critical points of causal variational
  principles}, arXiv:1703.05059 [math-ph], Adv. Theor. Math. Phys. \textbf{24}
  (2020), no.~3, 563--619.

\bibitem{review}
F.~Finster and M.~Jokel, \emph{Causal fermion systems: An elementary
  introduction to physical ideas and mathematical concepts}, arXiv:1908.08451
  [math-ph], {P}rogress and {V}isions in {Q}uantum {T}heory in {V}iew of
  {G}ravity (F.~Finster, D.~Giulini, J.~Kleiner, and J.~Tolksdorf, eds.),
  Birkh\"auser Verlag, Basel, 2020, pp.~63--92.

\bibitem{fockbosonic}
F.~Finster and N.~Kamran, \emph{Complex structures on jet spaces and bosonic
  {F}ock space dynamics for causal variational principles}, arXiv:1808.03177
  [math-ph], Pure Appl. Math. Q. \textbf{17} (2021), no.~1, 55--140.

\bibitem{gaugefix}
F.~Finster and S.~Kindermann, \emph{A gauge fixing procedure for causal fermion
  systems}, arXiv:1908.08445 [math-ph], J. Math. Phys. \textbf{61} (2020),
  no.~8, 082301.

\bibitem{dice2014}
F.~Finster and J.~Kleiner, \emph{Causal fermion systems as a candidate for a
  unified physical theory}, arXiv:1502.03587 [math-ph], J. Phys.: Conf. Ser.
  \textbf{626} (2015), 012020.

\bibitem{jet}
\bysame, \emph{A {H}amiltonian formulation of causal variational principles},
  arXiv:1612.07192 [math-ph], Calc. Var. Partial Differential Equations
  \textbf{56:73} (2017), no.~3, 33.

\bibitem{noncompact}
F.~Finster and C.~Langer, \emph{Causal variational principles in the
  $\sigma$-locally compact setting: {E}xistence of minimizers},
  arXiv:2002.04412 [math-ph], to appear in Adv. Calc. Var. (2021).

\bibitem{banach}
F.~Finster and M.~Lottner, \emph{Banach manifold structure and
  infinite-dimensional analysis for causal fermion systems}, arXiv:2101.11908
  [math-ph], Ann. Global Anal. Geom. \textbf{60} (2021), no.~2, 313--354.

\bibitem{pmt}
F.~Finster and A.~Platzer, \emph{A positive mass theorem for static causal
  fermion systems}, arXiv:1912.12995 [math-ph], to appear in Adv. Theor. Math.
  Phys. (2022).

\bibitem{support}
F.~Finster and D.~Schiefeneder, \emph{On the support of minimizers of causal
  variational principles}, arXiv:1012.1589 [math-ph], Arch. Ration. Mech. Anal.
  \textbf{210} (2013), no.~2, 321--364.

\bibitem{intro}
F.~Finster and J.-H. Treude, \emph{An {I}ntroductory {C}ourse on {C}ausal
  {F}ermion {S}ystems}, in preparation,
  \href{https://causal-fermion-system.com/intro-public.pdf}{www.causal-fermion-system.com/intro-public.pdf}.

\bibitem{glimm+jaffe}
J.~Glimm and A.~Jaffe, \emph{Quantum {P}hysics, a functional integral point of
  view}, second ed., Springer-Verlag, New York, 1987.

\bibitem{helgason}
S.~Helgason, \emph{Groups and {G}eometric {A}nalysis}, Mathematical Surveys and
  Monographs, vol.~83, American Mathematical Society, Providence, RI, 2000,
  Integral geometry, invariant differential operators, and spherical functions,
  Corrected reprint of the 1984 original.

\bibitem{langDG}
S.~Lang, \emph{Fundamentals of {D}ifferential {G}eometry}, Graduate Texts in
  Mathematics, vol. 191, Springer-Verlag, New York, 1999.

\bibitem{reed+simon}
M.~Reed and B.~Simon, \emph{Methods of {M}odern {M}athematical {P}hysics. {I},
  {F}unctional analysis}, second ed., Academic Press Inc., New York, 1980.

\end{thebibliography}
\providecommand{\bysame}{\leavevmode\hbox to3em{\hrulefill}\thinspace}
\providecommand{\MR}{\relax\ifhmode\unskip\space\fi MR }
\providecommand{\MRhref}[2]{%
  \href{http://www.ams.org/mathscinet-getitem?mr=#1}{#2}
}
\providecommand{\href}[2]{#2}

\end{document}